\documentclass{amsart}

\usepackage{amsmath,amsthm,amssymb,latexsym %,mathrsfs
}
\usepackage{graphicx,color,array,fullpage,fancyhdr,cancel}
\usepackage{hyperref,tabularx}
\usepackage{url}
\usepackage{tikz}
\usepackage{enumerate}
\usepackage{soul}
\usepackage{bbm}

\newtheorem{theorem}{Theorem}[section]
\newtheorem{lemma}[theorem]{Lemma}

\newtheorem{corollary}[theorem]{Corollary}

\theoremstyle{definition}

\theoremstyle{remark}
\newtheorem{remark}[theorem]{Remark}

\numberwithin{equation}{section}

\begin{document}

\title{Power dissipation in fractal AC circuits}

%\author{Eric Akkermans}\email{eric@physics.technion.ac.il}
\author{Joe P. Chen}\email{jpchen@colgate.edu}
\author{Luke G. Rogers}\email{rogers@math.uconn.edu}
\author{Loren Anderson}\email{and05097@umn.edu}
\author{Ulysses Andrews}\email{ulysses.andrews@uconn.edu}
%\author{Edith  Aromando}\email{edy32@wildcats.unh.edu}
\author{Antoni Brzoska}\email{antoni.brzoska@uconn.edu}
\author{Aubrey Coffey}\email{acoffey@agnesscott.edu}
\author{Hannah Davis}\email{davi2495@umn.edu}
%\author{Gerald Dunne }\email{dunne@phys.uconn.edu}
%\author{Michael Dworken}\email{mdworken@gmail.com}
\author{Lee~Fisher}\email{fisherla1@email.appstate.edu}
\author{Madeline  Hansalik}\email{mahans95@tamu.edu}
\author{Stephen Loew}\email{scloew@coe.edu}
\author{Alexander Teplyaev}\email{teplyaev@uconn.edu}

\urladdr{\url{http://math.colgate.edu/~jpchen/}}
\urladdr{\url{http://mathreu.uconn.edu/person/luke-rogers/}}
\urladdr{\url{http://teplyaev.math.uconn.edu/}}

\thanks{Research supported in part by NSF grants DMS-1106982, DMS-1262929 and DMS-1613025. 
%DOE grant ER41989, and Israel Science Foundation Grant No. 1232/13
}

\address{Department of Mathematics,
Colgate University,
Hamilton, NY, 13346
USA}
\address{Department of Mathematics,
 University of Connecticut, Storrs CT 06269, USA}
%\address{ Department of Physics, Technion - Israel Institute of Technology, 32000 Haifa, Israel}

\date{\today}

\begin{abstract}
We extend Feynman's analysis of an infinite ladder circuit to fractal circuits, providing examples in which fractal circuits constructed with purely imaginary impedances can have characteristic impedances with positive real part.  Using (weak) self-similarity of our fractal structures, we provide algorithms for studying the equilibrium distribution of energy on these circuits. This extends the analysis of self-similar 
resistance networks introduced by Fukushima, Kigami, Kusuoka, and more recently studied by Strichartz et al. 
\end{abstract}

\maketitle

\section{Introduction}\label{sec-intro}
Our goal in this paper is to generalize  Feynman's well-known example~\cite{F} of  an infinite ladder network of capacitors and inductors that exhibits a non-zero real resistance, despite having non-resistive components, to certain fractal networks.   
This extends the analysis of self-similar 
resistance networks introduced by Fukushima, Kigami, Kusuoka, and more recently studied by Strichartz et al 
\cite{akkermans2009physical,
	akkermans2010thermodynamics,
	ben1999not,
	fan2009harmonic,
	fukushima1992dirichlet,
	Kigamibook,
	Str06,
	strichartz2009fractal}. 
We also extend the  concepts of fractafolds and fractal quantum graphs 
\cite{alonso2016energy,
	chen2015spectral,
	lancia2015density,
	strichartz2003fractafolds,
	strichartz2012spectral} 
to the alternating current networks. 
The present paper is a step in the long program of transferring  the 
scalar analysis of the direct electrical currents on fractals  
to the various phenomena of wave propagation, such as the alternating current propagation (for which currently there are few mathematical physics references, 
\cite{akkermans1988theoretical,akkermans2013spontaneous,CT16,kusuoka1998waves}).  
In particular, obtaining the energy measure   
and 
resolvent estimates 
\cite{ionescu2010resolvent,
	rogers2012estimates}
and  studying effects of the magnetic field 
\cite{hinz2013dirac,hinz2015magnetic,hyde2016magnetic}
are among the more immediate goals of this research program. 

Feynman's ladder is an important example for understanding  wave propagation in AC circuits, so we expect our examples to be of interest for studying wave propagation on fractals, which is a topic of ongoing research~\cite{AkkDunnLevy}; in particular, we hope this model may be used to study fractal analogues of the telegraph equation.  In addition, our construction generalizes Kigami's method~\cite{Kigamibook} for obtaining Dirichlet forms on fractal sets as limits of finite resistance networks on graphs that approximate the fractal, as it permits us to consider finite networks of capacitors and inductors rather than resistors, and still obtain a non-trivial limit.  Kigami's method is central to the theory of analysis on fractals, and it was not previously known whether there was an extension of this theory to limits of electrical networks involving components other than resistors, though we learned while writing up this work that some circuits of this type were previously studied in~\cite{Barnsleyetal,dragon} and via continued fraction techniques in~\cite{Lak}.

One of the simplest non-trivial examples of a fractal network supporting an interesting Dirichlet form is the Sierpinski Gasket (SG), see~\cite{Kusuoka85,BarlowPerkins,Kig89}.  The space of such forms is quite complicated: if no symmetry is assumed, then there is a condition due to Sabot~\cite{Sabot} that gives existence and uniqueness, but it is difficult to give an explicit description of such forms (see~\cite{MRT} for one approach to this). Nevertheless, we show in Section~\ref{sec-SG1} that there are no bilaterally symmetric examples of the type we seek on SG.  For this reason we introduce a  new set, the Feynman-Sierpinski Ladder, and also construct two different circuits based on the Hanoi attractor~\cite{Sunic,NT}.  These sets are not self-similar in the usual sense (see, e.g., Chapter 1 of~\cite{Kigamibook}), but are of a type sometimes called \emph{weakly} self-similar.  They fall within the well-known Mauldin-Williams class~\cite{MW}, for which the limits of resistor networks are considered in~\cite{HamblyNyberg,Metz}.  They are described using a Mauldin-Williams type replacement rule, but with a separate factor to describe the rescaling of impedances from level to level.  In the case of the Hanoi attractor, this is related to the fractal quantum graph approach in~\cite{2012hanoi,alonso2016energy}.

Like Feynman's ladder, each of our fractal circuits is built from purely imaginary impedances (inductors and capacitors).  For each we give conditions on the scaling factors that imply the resulting circuit is a filter, meaning that its characteristic impedance (or impedances) have positive real part.  We also give algorithms for computing the equilibrium voltages in the circuit if an AC signal is applied to the boundary points.  Mathematically this latter is a harmonic function, and we give the value of any harmonic function at each point of the circuit using an infinite product of matrices derived from the (weak) self-similar structure on the set.  This may be seen as a harmonic interpolation method, in which the values of the harmonic function on each level of the construction give those on the next level.

%  some recent papers on physics on fractals \cite{ADT09,ADT10,ABDTV12,Akk,Dunne12}  
% mesoscopic physics, \cite{AM}, 
% early physics papers on fractals  \cite{AO,DABK,GAA,RT,sgmagnetic}.  
% The singularity of energy measures  in related real circuits is proved in   \cite{BST99,hino}.  
% Hanoi-type graphs play an important rile in \cite{MRT,FST}, which is closely related to  \cite{englert}. 
% energy in fractal resistance networks  \cite{Ngasket,fractalina,twists2,HMT,BBKT,grad,Tcjm,PT,OSCT}

\section{Self-similar AC circuits on the Sierpinski  Gasket}\label{sec-SG1}
The existence of non-trivial resistance forms (as defined in~\cite{KigJFA03}) on the Sierpinski Gasket (SG) is well-known, as is the fact that these correspond to limits of physical configurations of resistors on graph approxmations to the fractal.  It is then natural to ask whether there are self-similar configurations of electrical elements for which the impedance has non-zero imaginary part, so that the corresponding  objects should have interesting AC behavior.  Here we show that there are no non-trivial examples with bilateral symmetry, as the self-similarity implies the only possiblities are complex multiples of the usual resistance forms on SG. 

Recall that as a subset of $\mathbb{R}^2$ the SG may be defined to be the (unique non-empty compact) invariant subset of a collection of three contractive maps $F_j(x)=\frac{1}{2}(x-p_j)+p_j$, $j=0,1,2$, where the points $p_j$ are the vertices of an equilateral triangle.  A circuit structure on the SG is self-similar if there are scalings $r_0,r_1,r_2$ such that the resistance  between $F_j(x)$ and $F_j(y)$ in the subcircuit $F_j(\text{SG})$ (i.e.\@ the circuit with the other two cells removed) is $r_j$ times the resistance between $x$ and $y$ in SG.  
\begin{theorem}
Any self-similar bilaterally symmetric AC circuit on the Sierpinski gasket is a constant complex multiple of a purely real self-similar circuit.
\end{theorem}
\begin{proof}
Induction on the self-similarity condition shows that the internal impedances on any cell are obtained by multiplying those on the gasket by a factor $r_1^{n_1}r_2^{n_2}r_3^{n_3}$, where the integers $\{n_j\}_{j=1}^3$ are in one-to-one correspondence with the location of the cell.  Since
the resulting impedances must have non-negative real part, it follows that all $r_j$ are real and positive.  The assumption of bilateral symmetry is equivalent to assuming $r_2=r_3=r$.  Calculations in \cite{Sabot} show that if $r_1<3r/2$, then there is a unique Dirichlet form with this self-similar scaling; but that for $r_1\geq r$, the form degenerates onto a strict subset of SG.  This can also be verified by direct calculation as in~\cite{MRT}, see also Chapter~4 of~\cite{Str06}. With $s=r_1/r$ (which must be real and positive), one may compute the conductances on the level-$0$ graph to be $1$ and
\begin{equation*}
	\frac{s^2-1 + \sqrt{(s^2-1)^2+s^2(3-2s)}}{3-2s}
	\end{equation*}
up to a constant factor. It is readily verified that for real positive $s$ this latter expression is real.  Any self-similar circuit with these scalings  is therefore a constant complex multiple of this purely real resistance circuit.
\end{proof}

\begin{remark}
We expect that the above result holds without the assumption of bilateral symmetry, but have not tried to prove this.  
The work of Sabot~\cite{Sabot} mentioned above completely solves the question of existence and uniqueness on SG of self-similar circuits built from real resistances, giving  necessary and sufficient conditions on the scalings $r_j$. However, unlike in the bilaterally symmetric case, the solution is not explicit, but is given as a fixed point of a non-linear map that is contractive in a suitable metric on a cone in a real projective space.  It does not seem to follow directly that there cannot be complex solutions other than complex multiples of the unique real solution.  We believe it might be possible to prove this by considering a map analogous to that studied by Sabot on a suitable complex projective space. 
\end{remark}

\section{Feynman-Sierpinski Ladder Circuit}

By analogy with the classical ladder circuit of Feynman~\cite{F}, we construct a Feynman-Sierpinski Ladder (FSL) following the substitution procedure shown in Figure~\ref{fig-SGLadder}.  Three copies of the triangular element shown at the left are glued (with boundary points identified) and connected inside the triangle using  capacitors of impedance $Z_C$ and  inductors of impedance $Z_L$, as indicated in the central image.  This process is iterated infinitely many times.  The second step of the iteration is shown at the right.  Let $\frac{2}{3}Z$ denote the characteristic impedance, i.e.\ the impedance across two external vertices in the limiting structure; this value is chosen because it would be the effective impedance if there were impedance $Z$ across each edge of the initial triangle.   We show that each choice of capacitance and inductance determines a unique such $Z$, and gives conditions under which $Z$ has postive real part,  meaning that the circuit is a filter.

\begin{figure}[htb]
  \centering
  \def\svgwidth{\columnwidth}
    \resizebox{\textwidth}{!}{
\begingroup%
  \makeatletter%
  \providecommand\color[2][]{%
    \errmessage{(Inkscape) Color is used for the text in Inkscape, but the package 'color.sty' is not loaded}%
    \renewcommand\color[2][]{}%
  }%
  \providecommand\transparent[1]{%
    \errmessage{(Inkscape) Transparency is used (non-zero) for the text in Inkscape, but the package 'transparent.sty' is not loaded}%
    \renewcommand\transparent[1]{}%
  }%
  \providecommand\rotatebox[2]{#2}%
  \ifx\svgwidth\undefined%
    \setlength{\unitlength}{583.96615175bp}%
    \ifx\svgscale\undefined%
      \relax%
    \else%
      \setlength{\unitlength}{\unitlength * \real{\svgscale}}%
    \fi%
  \else%
    \setlength{\unitlength}{\svgwidth}%
  \fi%
  \global\let\svgwidth\undefined%
  \global\let\svgscale\undefined%
  \makeatother%
  \begin{picture}(1,0.28481908)%
    \put(0,0){\includegraphics[width=\unitlength,page=1]{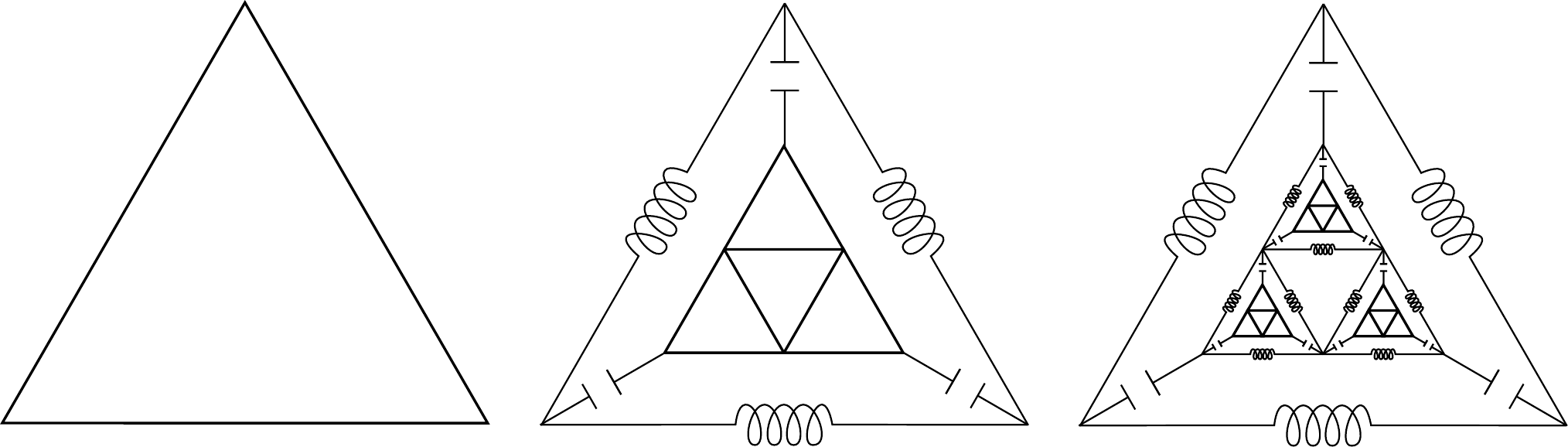}}%
    \put(0.02884536,0.03){\color[rgb]{0,0,0}\makebox(0,0)[lb]{\smash{$p_1$}}}%
    \put(0.148,0.248){\color[rgb]{0,0,0}\makebox(0,0)[lb]{\smash{$p_0$}}}%
    \put(0.27,0.03){\color[rgb]{0,0,0}\makebox(0,0)[lb]{\smash{$p_2$}}}%
    \put(0.42475872,0.04){\color[rgb]{0,0,0}\makebox(0,0)[lb]{\smash{$q_1$}}}%
    \put(0.50969515,0.19){\color[rgb]{0,0,0}\makebox(0,0)[lb]{\smash{$q_0$}}}%
    \put(0.56,0.04){\color[rgb]{0,0,0}\makebox(0,0)[lb]{\smash{$q_2$}}}%
    \put(-0.00403326,-0.16695623){\color[rgb]{0,0,0}\makebox(0,0)[lt]{\begin{minipage}{0.77949723\unitlength}\raggedright \end{minipage}}}%
  \end{picture}%
\endgroup%
}
\caption{Feynman-Sierpinski Ladder circuit construction.}\label{fig-SGLadder}
\end{figure}
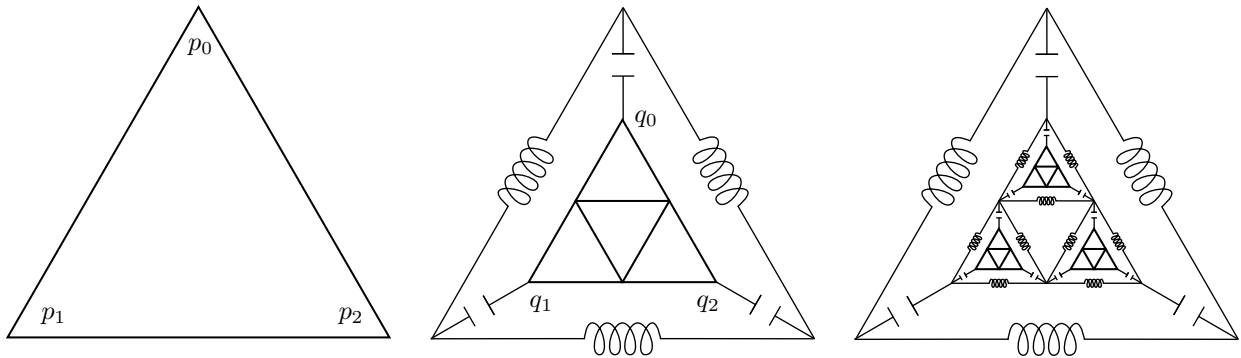 

\begin{theorem}\label{sgladderfilterthm}
If $Z_C$ and $Z_L$ are fixed, and $\frac{2}{3}Z$ is the impedance between any two vertices of the limiting structure, then $Z$ satisfies
\begin{equation}\label{sgladdereqn1}
	\frac{1}{Z} = \frac{1}{Z_L} + \frac{1}{3Z_C+5Z/3}.
	\end{equation}
If the capacitances are $C$, the inductances $L$, and the applied AC signal has frequency $\omega$  (so that $Z_C=\frac{1}{i\omega C}$ and $Z_L=i\omega L$), then there is a solution in which $Z$ has positive real part (i.e. the circuit is a filter) precisely when
\begin{gather}\label{sgladderfiltercondit}
	9(4 - \sqrt{15 })< 2\omega^2 LC < 9(4 + \sqrt{15}),\\
\intertext{in which case the impedance is}
Z =\frac{1}{10\omega C} \biggl(  (9+ 2\omega^2 LC)i + \sqrt{ 144\omega ^2 LC -  4(\omega ^2 LC)^2  -  81} \biggr). \label{sgladderimped}
\end{gather}
Outside the frequency range \eqref{sgladderfiltercondit}, $Z$ is purely imaginary and is given by
\begin{equation*}
	Z=\begin{cases}
	\frac i{10\omega C}\left(2\omega^2LC +9- \sqrt{4(\omega^2LC)^2+81-144\omega^2LC}\right) &\text{ if }\quad 2\omega^2LC<9(4-\sqrt{15}),\\
	\frac i{10\omega C}\left(2\omega^2LC +9+ \sqrt{4(\omega^2LC)^2+81-144\omega^2LC}\right) &\text{ if  }\quad 2\omega^2LC> 9(4+\sqrt{15}).
	\end{cases}
	\end{equation*}
\end{theorem}
\begin{proof}
In the central image in Figure~\ref{fig-SGLadder} the internal structure made from three triangles with sides of impedance $Z$ is the first step in constructing a Sierpinski Gasket. It is well-known that this is equivalent to a triangle with edge impedances $5Z/3$ or an inverted Y-circuit with edge impedances $5Z/9$. Including the capacitors gives an inverted Y-circuit of edge impedance $Z_C+5Z/9$, which is equivalent to a triangular circuit of edge impedance $3Z_C+5Z/3$.  Each edge is in parallel with an inductor, so the final impedance of the edge is as stated in equation~\eqref{sgladdereqn1}.

Formally solving~\eqref{sgladdereqn1}  produces~\eqref{sgladderimped}  with both positive and negative branches of the square root. The condition~\eqref{sgladderfiltercondit} arises from requiring the quantity under the square root to be positive, as is necessary for the circuit to be a filter. In order to determine the correct branches for the other regions, we consider the cases $\omega\to0$ and $\omega\to\infty$. In the former $Z\to0$ because the inductors have vanishing impedance, while in the latter the capacitors have vanishing impedance, and one may check from~\eqref{sgladdereqn1} that $Z\to\frac{2}{5}Z_{L}$. %JPC: I changed $\sim$ to $\to$
\end{proof}

\subsection*{Convergence of finite approximations}

It is known~\cite{Yoon} in the case of Feynman's ladder that the characteristic impedances of the sequence of finite approximations to the ladder fail to converge to the impedance of the infinite circuit. Introducing a small positive resistance $\epsilon$ in series with each of  the capacitors and inductors gives a sequence of scale $N$ approximating circuits for which the characteristic impedances  $Z_{N,\epsilon}$ do converge as $N\to\infty$, and moreover, the regularized limit $\lim_{\epsilon\to0+}\lim_{N\to\infty}Z_{N,\epsilon}$ is the impedance of the infinite ladder. The following theorem establishes the same result for the Feynman-Sierpinski Ladder.  Its proof is closely related to results in~\cite{Barnsleyetal}.

\begin{figure}[htb]
  \centering
  \def\svgwidth{\columnwidth}
    \resizebox{\textwidth}{!}{\includegraphics{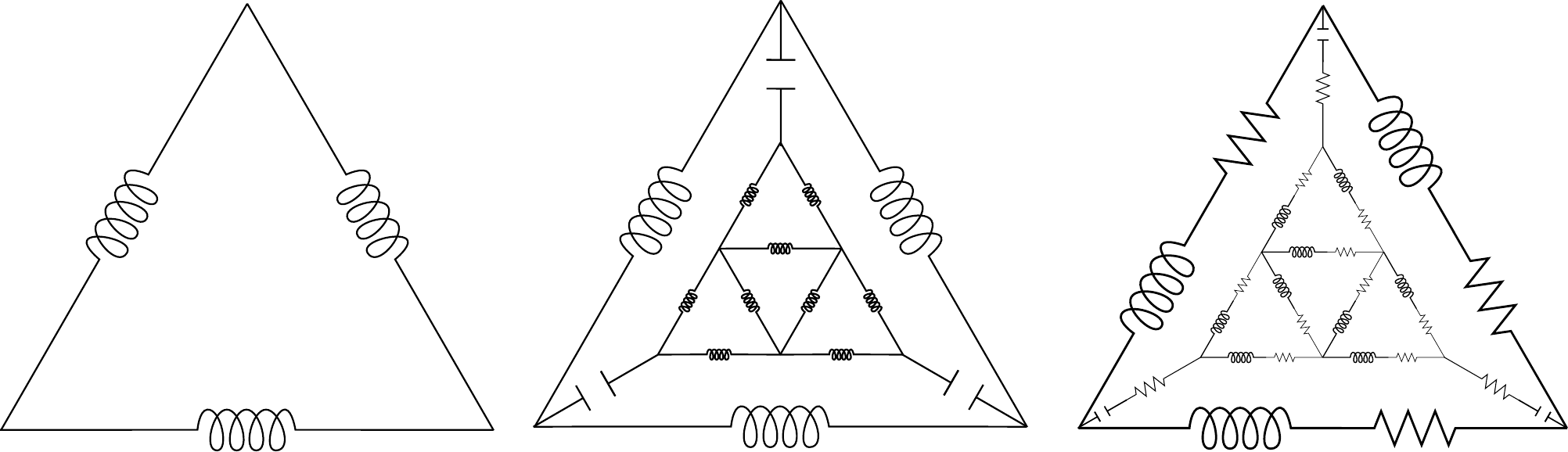}}
\caption{Left and center: two finite approximations to the Feynman-Sierpinski Ladder. Right: the center approximation with resistances in series.}
\label{fig-SGLadderNsteps}
\end{figure} 

\begin{theorem}
Consider the Feynman-Sierpinski Ladder under a filter condition, and let $\frac{2}{3}Z$ be the impedance across two nodes.
\begin{enumerate}[(a)]
\item Let $\frac{2}{3}Z_N$ be the  impedance across two external vertices of the $N^{\text{th}}$ approximation of this Feynman-Sierpinski Ladder, so the left and central diagrams in Figure~\ref{fig-SGLadderNsteps} have impedances $\frac{2}{3}Z_0$ and $\frac{2}{3}Z_1$.  Then the sequence $Z_N$ is not convergent.
\item  For $\epsilon>0$, let $\frac{2}{3}Z_{N,\epsilon}$ be the impedance of the  $N^{\text{th}}$ approximation of the Feynman-Sierpinski Ladder in which each capacitor and inductor is in series with a resistor of resistance $\epsilon$. (The diagram for $N=1$ is on the right of Figure~\ref{fig-SGLadderNsteps}.)  Then  $\lim_{N\to\infty} Z_{N,\epsilon}$ exists for all $\epsilon\in(0,\infty)$ and $\lim_{\epsilon\to0+}\lim_{N\to\infty}Z_{N,\epsilon}=Z$.
\end{enumerate}
\end{theorem}
\begin{proof}
Observe that $Z_{1}$ is obtained from $Z_0$ by applying a fractional linear transformation (FLT) which we denote $F$. From the argument leading to~\eqref{sgladdereqn1} we see that
\begin{equation*}
	Z_{1} =  \biggl( \frac{1}{Z_L} + \frac{1}{3Z_C+5Z_0/3}\biggr)^{-1}
	=\frac{ Z_L \bigl( 3Z_C + 5Z_0/3\bigr) } { Z_L + 3Z_C + 5Z_0/3 }
	= \frac{5Z_L Z_0 + 9Z_L Z_C } { 5Z_0 + 3Z_L + 9 Z_C  }  = F(Z_0),
	\end{equation*}
where $F$ depends on $Z_L$ and $Z_C$.  However, the same argument shows that  if $Z_N=G(Z_0)$ then $Z_{N+1}=G(F(Z_0))$, so by induction $Z_N=F^{\circ N}(Z_0)$, where $F^{\circ N}$ denotes $N$-fold composition.

FLTs have very simple dynamics.  Notice that if $F$ has two fixed points $Z_{\pm}$ (i.e.\ there are two  solutions to the quadratic $Z=F(Z)$), and $H$ is the  FLT taking $Z_+$ to $0$ and $Z_-$ to $\infty$, then the conjugated transformation $H\circ F\circ H^{-1}$ is an FLT with fixed points at $0$ and $\infty$. By a local holomorphic change of coordinates, this FLT can be made into a linear map, i.e.\ multiplication by a fixed complex number. This number is obviously the derivative $\bigl(H\circ F\circ H^{-1}\bigr)'(0)=F'(Z_+)$.  Thus we see that if $|F'(Z_+)|<1$, then $Z_+$ attracts all points except $Z_-$; and if $|F(Z_+)|>1$, then $Z_-$ attracts all points except $Z_+$; and if $|F'(Z_+)|=1$, neither of $Z_\pm$ attracts any other orbit.  

Since we assumed that the Feynman-Sierpinski Ladder was a filter (i.e.\@~\eqref{sgladderfiltercondit} holds), its FLT has two fixed points.  It remains to compute the derivative at the relevant one, which, as in~\eqref{sgladderimped} is
\begin{equation}\label{sgladderimpedinZCZL}
	Z  = \frac{1}{10} \Bigl( -9Z_C+ 2Z_L +  \sqrt{(9Z_C+8Z_L)^2  - 60Z_L^2 } \Bigr).
\end{equation}
Then
\begin{align}
	F'(Z) = \frac{15Z_L^2} {(5Z + 3Z_L +9Z_C)^2} 
	&= \frac{ (9Z_C + 8Z_L)^2 - \bigl( (9Z_C+8Z_L)^2 - 60 Z_L^2\bigr) }{ \bigl( (9Z_C + 8Z_L)+ \sqrt{(9Z_C+8Z_L)^2  - 60Z_L^2 } \bigr)^2}\notag\\
	&= \frac{ 9Z_C + 8Z_L - \sqrt{(9Z_C+8Z_L)^2- 60Z_L^2 }}{ 9Z_C + 8Z_L + \sqrt{(9Z_C+8Z_L)^2- 60Z_L^2 }}. \label{FLTderiv}
	\end{align}
In terms of the capacitance, inductance and frequency as in Theorem~\ref{sgladderfilterthm}, we then have under the filter condition~\eqref{sgladderfiltercondit} that
\begin{equation*}
	F'(Z) = \frac{ (9- 8\omega^2 LC)i - \sqrt{ 144\omega ^2 LC -  4(\omega ^2 LC)^2  -  81} } { (9- 8\omega^2 LC)i + \sqrt{ 144\omega ^2 LC -  4(\omega ^2 LC)^2  -  81}},
	\end{equation*}
so $|F'(Z)|=1$ because the quantity under the root is real.  This and our discussion of the dynamics establishes that $Z_N= F^{\circ N}(Z_0)$ does not converge.

To instead consider $Z_{N,\epsilon}$ one can increment $Z_C$ and $Z_L$ by $\epsilon\in(0,\infty)$ in~\eqref{FLTderiv}; call the perturbed map $F_\epsilon$ and its fixed point $Z_\epsilon$. Note that the square root was chosen in~\eqref{sgladderimped} to be the physical one in which the circuit impedance also has positive real part, so it is the branch that is continuous in $\mathbb{C}\setminus(-\infty,0]$.  We wish to show that $\bigl|F'_\epsilon(Z_\epsilon)\bigr|<1$.  It is convenient to conformally map the unit disc to the right half plane and instead show the equivalent statement that the real part of
\begin{equation}\label{sgladderconvergenceeqn}
	\frac{ 1+F'_\epsilon(Z_\epsilon) }{1- F'_\epsilon(Z_\epsilon) } 
	=\frac{9Z_C + 8Z_L+17\epsilon} { \sqrt{(9Z_C + 8Z_L+17\epsilon)^2 -60(Z_L+\epsilon)^2 } }.
	\end{equation}
is always positive.
Observe in particular that as $\epsilon\to\infty$ on the positive real axis the expression~\eqref{sgladderconvergenceeqn} tends to $17/\sqrt{17^2-60}$, which has positive real part.  By connectedness of $(0,\infty)$ and continuity of~\eqref{sgladderconvergenceeqn} in $\epsilon$, it then suffices to show that the real part of~\eqref{sgladderconvergenceeqn} is never zero.  Note that continuity in $\epsilon$ follows from our choice of square root and the fact that the expression under the root always has positive real part from the filter condition and positivity of $\epsilon$.

If~\eqref{sgladderconvergenceeqn} were to have vanishing real part, then its square would be negative and real, as would the reciprocal of the square, which may be written as
\begin{equation*}
	1- \frac{60(Z_L+\epsilon)^2}{(9Z_C + 8Z_L+17\epsilon)^2 }.
	\end{equation*}
But then $(Z_L+\epsilon)/(9Z_C + 8Z_L+17\epsilon)\in\mathbb{R}$, whence so is $(Z_L+\epsilon)(\overline{9Z_C + 8Z_L+17\epsilon})$.  The imaginary component of the latter is $\epsilon(9Z_L-9Z_C)$, so this is only possible if $\epsilon=0$ or $Z_L=Z_C$.  However $Z_L=i\omega L$ with $L>0$ and $Z_C=1/i\omega C$ with $C>0$, so the latter case does not occur.  We conclude that if $\epsilon\in(0,\infty)$ then  $|F'_\epsilon(Z_\epsilon)|<1$, so it is an attractive fixed point and $\lim_N Z_{N,\epsilon}=\lim_N F_\epsilon^{\circ N}(Z_{0,\epsilon})=Z_\epsilon$. Moreover, $\lim_{\epsilon\to0^+} \lim_{N\to\infty} Z_{N,\epsilon}$ is then just $\lim_{\epsilon\to0^+} Z_\epsilon$.  Since $Z_\epsilon$ is given by incrementing $Z_C$ and $Z_L$ in~\eqref{sgladderimpedinZCZL} by $\epsilon$, and the square root is continous in a neighborhood of the positive real axis, $Z_\epsilon$ is continuous in $\epsilon$, which completes the proof.
\end{proof}

\subsection*{Harmonic Functions} With a good understanding of the impedance structures that are possible on the Feynman-Sierpinski Ladder, we proceed to study the harmonic functions, which are equilibrium states of the circuit when a potential is applied across the boundary vertices.  The value of a harmonic function at a fixed point in the Feynman-Sierpinski Ladder depends linearly on the boundary values, so the space of harmonic functions is three dimensional.  The following theorem gives an efficient way to compute the value of a harmonic function from its boundary values by exploiting the weakly self-similar structure of the ladder.

\begin{theorem}
Label the vertices of FSL as in Figure~\ref{fig-SGLadder}. If we give the values of a harmonic function $h$ on the boundary vertices of FSL as a vector $v=(h(p_0),h(p_1),h(p_2))^T$, then the values at the next level are $(h(q_0),h(q_1),h(q_2))^T=Mv$, and those at the internal triangles are $M_0Mv$, $M_1Mv$ and $M_2Mv$, where the matrices are
\begin{gather}
	M= \frac{1}{9z_{C}+5z}  
	\begin{bmatrix}
	3z_{C}+5z & 3z_C & 3z_C \\
	3z_C & 3z_{C}+5z & 3z_C \\
	3z_C & 3z_C &3z_{C}+5z 
	\end{bmatrix}, \label{eqn:matrixMforSGL}\\
	M_{0}=
	\begin{bmatrix}
	1 & 0 &0 \\
	\frac{2}{5} & \frac{2}{5} & \frac{1}{5} \\
	\frac{2}{5} & \frac{1}{5} & \frac{2}{5}
	\end{bmatrix}, \quad
	M_{1}=\begin{bmatrix}
	\frac{2}{5} & \frac{2}{5} & \frac{1}{5} \\
	0 & 1 & 0 \\
	\frac{1}{5} & \frac{2}{5} & \frac{2}{5}
	\end{bmatrix}, \quad
	M_{2}=\begin{bmatrix}
	\frac{2}{5} & \frac{1}{5} & \frac{2}{5}\\
	\frac{1}{5} & \frac{2}{5} & \frac{2}{5} \\
	0 & 0 & 1
	\end{bmatrix}. \label{eqn:othermatricesforSGL}
	\end{gather}

\end{theorem}
\begin{proof}
As stated in the proof of Theorem~\ref{sgladderfilterthm}, the three small copies of the circuit shown as triangles in the middle of Figure~\ref{fig-SGLadder} are equivalent to a triangular circuit with vertices the points $q_j$ and edge impedances $5Z/3$.  The fact that the currents at each $q_j$ sum to zero may then be written as a matrix equation
\begin{equation*}
	\frac{3}{5Z}\begin{bmatrix}
		2 & -1 & -1 \\
		-1& 2 & -1\\
		-1 &-1& 2
		\end{bmatrix}	
	\begin{bmatrix}
		V(q_{0})\\
		V(q_{1})\\
		V(q_{2})
		\end{bmatrix}
	+\frac{1}{Z_C} 
	\begin{bmatrix}
	V(q_{0})-V(p_0)\\
	V(q_{1})-V(p_1)\\
	V(q_{2})-V(p_2)
	\end{bmatrix}
	=0,
\end{equation*}
for which the solution is $(V(q_0),V(q_1),V(q_2))^T=M(V(p_0),V(p_1),V(p_2))^T$.

The fact that the values on the vertices of the three small copies of the circuit are given by applying the matrices $M_j$, $j=0,1,2$, to the vector of values at the points $q_j$ is then a standard result in the electrical network theory for the Sierpinski Gasket; see~\cite[Chapter~1]{Str06} for a proof and references to the original literature.
\end{proof}

Power dissipation in fractal Feynman-Sierpinski AC circuit is studied in 
\cite{AC3}.

\section{Hanoi Circuit I}

An alternative modification of the Sierpinski Gasket construction that yields interesting AC circuits is related to the Schreirer graph of the Hanoi Towers group (see~\cite{Sunic}), for which reason we refer to it as a Hanoi circuit.  The set on which it is defined, the Hanoi attractor, is shown in Figure~\ref{fig-Hanoi}. 

\begin{figure}[htb]
\centering
\includegraphics[width=6cm]{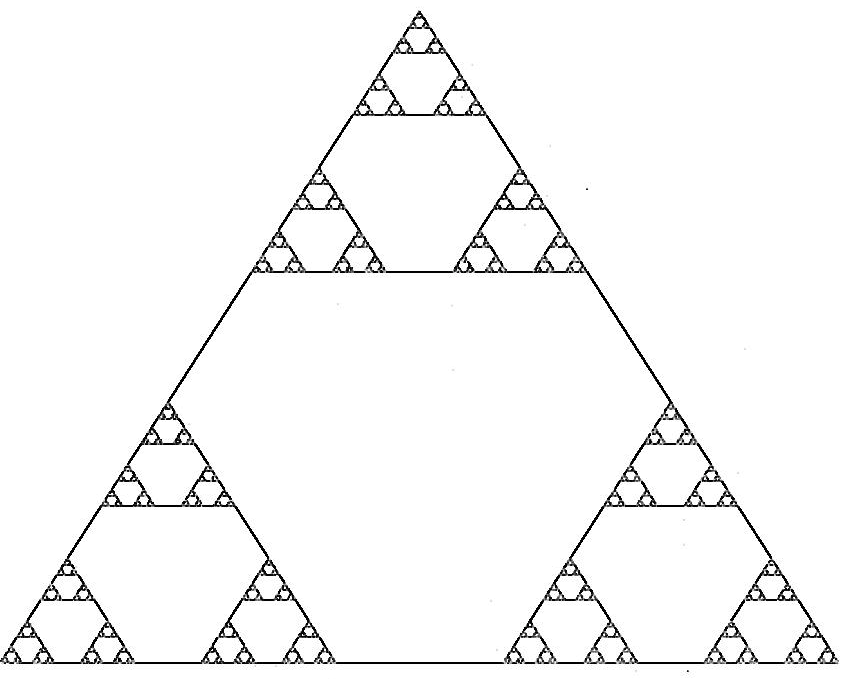}
\caption{The Hanoi Attractor}\label{fig-Hanoi}
\end{figure}

 Resistance forms, spectral properties of the associated Laplacians and other aspects of circuit theory on this set may be found in~\cite{Strichartzaverages,Sunic,2012hanoi,alonso2016energy}.  For our problem we follow the treatment in \cite{AC-fractal}, and begin the circuit construction with a  bilaterally symmetric inverted Y-circuit in which the vertical arm has impedance $Z_1$ and the angled arms each have impedance $Z_2$, as shown on the left in Figure~\ref{fig-Hanoi1}.
\begin{figure}[htb]
  \centering
  \def\svgwidth{\columnwidth}
    \resizebox{0.9\textwidth}{!}{
\begingroup%
  \makeatletter%
  \providecommand\color[2][]{%
    \errmessage{(Inkscape) Color is used for the text in Inkscape, but the package 'color.sty' is not loaded}%
    \renewcommand\color[2][]{}%
  }%
  \providecommand\transparent[1]{%
    \errmessage{(Inkscape) Transparency is used (non-zero) for the text in Inkscape, but the package 'transparent.sty' is not loaded}%
    \renewcommand\transparent[1]{}%
  }%
  \providecommand\rotatebox[2]{#2}%
  \ifx\svgwidth\undefined%
    \setlength{\unitlength}{568.01600159bp}%
    \ifx\svgscale\undefined%
      \relax%
    \else%
      \setlength{\unitlength}{\unitlength * \real{\svgscale}}%
    \fi%
  \else%
    \setlength{\unitlength}{\svgwidth}%
  \fi%
  \global\let\svgwidth\undefined%
  \global\let\svgscale\undefined%
  \makeatother%
  \begin{picture}(1,0.28552043)%
    \put(0,0){\includegraphics[width=\unitlength,page=1]{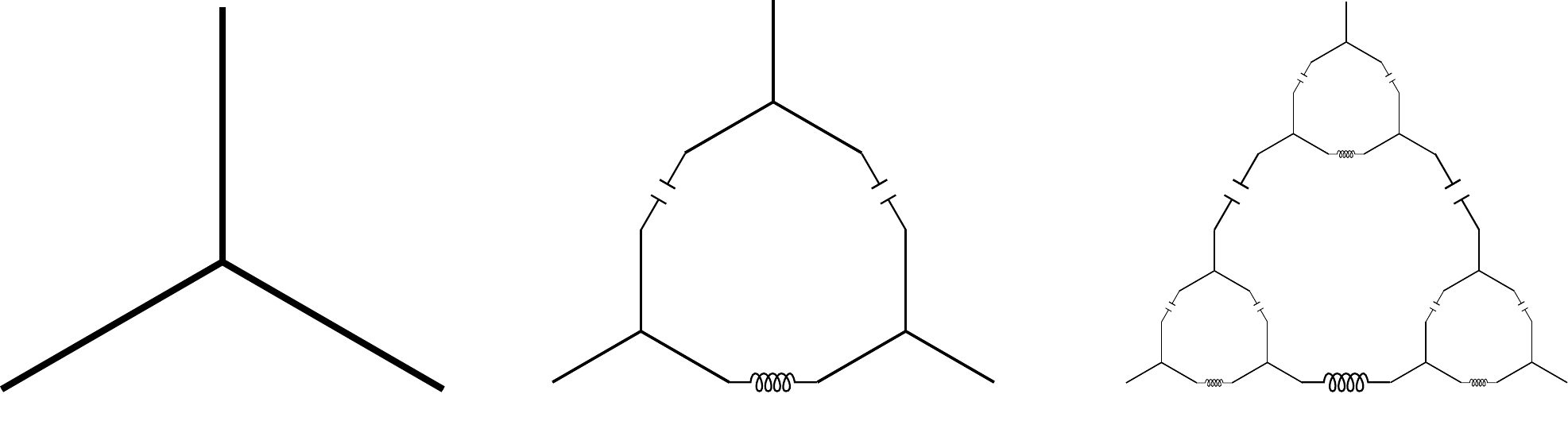}}%
    \put(0.15774204,0.21126166){\color[rgb]{0,0,0}\makebox(0,0)[lb]{\smash{}}}%
    \put(0.15223068,0.19219267){\color[rgb]{0,0,0}\makebox(0,0)[lb]{\smash{$Z_1$}}}%
    \put(0.21596808,0.09851801){\color[rgb]{0,0,0}\makebox(0,0)[lb]{\smash{$Z_2$}}}%
    \put(0.02674767,0.09752211){\color[rgb]{0,0,0}\makebox(0,0)[lb]{\smash{$Z_2$}}}%
    \put(0.5047782,0.25493418){\color[rgb]{0,0,0}\makebox(0,0)[lb]{\smash{$rZ_1$}}}%
    \put(0.42299885,0.21176885){\color[rgb]{0,0,0}\makebox(0,0)[lb]{\smash{$rZ_2$}}}%
    \put(0.57648282,0.16714619){\color[rgb]{0,0,0}\makebox(0,0)[lb]{\smash{$Z_C$}}}%
    \put(0.37508004,0.16714619){\color[rgb]{0,0,0}\makebox(0,0)[lb]{\smash{$Z_C$}}}%
    \put(0.48449339,0.00845047){\color[rgb]{0,0,0}\makebox(0,0)[lb]{\smash{$Z_L$}}}%
    \put(0.75394261,0.16714619){\color[rgb]{0,0,0}\makebox(0,0)[lb]{\smash{$Z_C$}}}%
    \put(0.84997606,0.00845047){\color[rgb]{0,0,0}\makebox(0,0)[lb]{\smash{$Z_L$}}}%
    \put(0.93774025,0.16714619){\color[rgb]{0,0,0}\makebox(0,0)[lb]{\smash{$Z_L$}}}%
  \end{picture}%
\endgroup%
}
\caption{A fractal Hanoi-type circuit construction.}\label{fig-Hanoi1}
\end{figure}
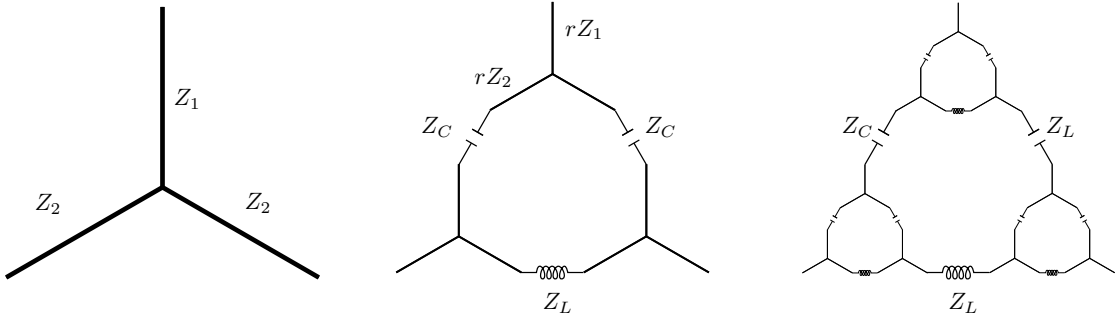 
We perform a replacement operation as shown to obtain a sequence of circuit graphs.  Each graph is obtained by rescaling the previous one and joining three copies as in the central diagram. Impedances on each copy are rescaled by a real factor $r>0$. 

When seeking conditions that imply the characteristic impedances of the first and second circuits are the same, symmetry suggests we should work with the impedance from the upper vertex to the two lower vertices, and that between the two lower vertices.  By Kirchoff's laws we obtain the following system of equations relating $Z_1$ and $Z_2$:
\begin{align}
\label{eq:1} Z_1 + \frac{1}{2}Z_2 &= rZ_1 + \frac{1}{2}\left(rZ_1 + 2r Z_2 + Z_C\right)\\
\label{eq:2} 2Z_2 &= 2r Z_2 + \left(\frac{1}{2rZ_2 + Z_L} + \frac{1}{2}\frac{1}{rZ_1 + rZ_2 + Z_C} \right)^{-1}.
\end{align}
It is convenient that these reduce to a quadratic for $Z_2$ and a formula for $Z_1$.  We state this as a lemma.

\begin{lemma}\label{HanoiIsolnsforZ}
If~\eqref{eq:1} and~\eqref{eq:2} hold, then $r\neq1$.  If also $r\neq\frac{2}{3}$, then
\begin{gather}
Z_1 = \frac{1}{2-3r}\bigl( (2r-1)Z_2 + Z_C\bigr), \text{ and}\label{eq:3}\\
r(5r-3) Z_2^2 + (2r-1)(2Z_C+Z_L) Z_2 + Z_C Z_L =0, \label{eq:4}
\end{gather}
while if $r=\frac{2}{3}$, then $Z_2=-3Z_C$ and $Z_1= \frac{3Z_C(Z_L+2Z_C)}{4(Z_L-2Z_C)}$.
\end{lemma}
\begin{proof}
The fact that $r\neq1$ is immediate from~\eqref{eq:2}, as otherwise the reciprocal on the right is zero.  The expression~\eqref{eq:3}  is immediate from~\eqref{eq:1}.  Also note from (\ref{eq:1}) that, provided $r\neq\frac{2}{3}$,
\begin{align}
\label{eq:1''}
r(Z_1 + Z_2) + Z_C = \frac{r}{2-3r}\bigl( (1-r)Z_2 + Z_C) +Z_C  = \frac{1-r}{2-3r} (rZ_2 +2 Z_C).
\end{align}
From \eqref{eq:1''} and \eqref{eq:2} we get
\begin{equation*}
\frac{1}{(1-r) Z_2 }= \frac{2}{2rZ_2 + Z_L} + \frac{2-3r}{(1-r) (rZ_2 +2 Z_C)},
\end{equation*}
so that
\begin{equation*}
\frac{(4r-2)Z_2 + 2 Z_C}{Z_2 (rZ_2+2Z_C)} = \frac{2(1-r)}{2rZ_2+Z_L}.
\end{equation*}
Simplifying the resulting quadratic in $Z_2$ gives~\eqref{eq:4}.

The remaining case $r=\frac{2}{3}$ has $Z_2=-3Z_C$ from~\eqref{eq:1}, and substituting in~\eqref{eq:2} gives 
\begin{equation*}
	\frac{1}{2Z_C} = \frac{1}{4Z_C- Z_L} + \frac{3}{9Z_C-4Z_1},
	\end{equation*}
from which we can compute the stated value for $Z_1$.
\end{proof}

While it is useful to know that the values of $r$, $Z_C$ and $Z_L$ determine at most two pairs of characteristic impedances $(Z_1,Z_2)$, it is not immediately apparent that this result has any physical meaning, because \emph{a priori} the impedances in the circuit could have negative real part.  Our next result identifies the conditions for the circuit to be a filter.  More specifically, we suppose that the circuit is driven by an AC signal with frequency $\omega$, write $C$ for the capacitance and $L$ for the inductance of the components used in the construction, so that $Z_C=\frac{1}{i\omega C}$ and $Z_L=i\omega L$, and identify those values of $r$ and $\omega$ for which the impedances that can be measured in the circuit  have  positive real part.  One point about which we must be careful is that $Z_1$ and $Z_2$ are not actually impedances in the circuit:  they are values computed from a $\Delta$-Y transform.  The basic circuit has three external terminals, and we can measure the impedance across any two. Thus by the leftmost diagram in Figure~\ref{fig-Hanoi1} we can measure $Z_1+Z_2$ and $2Z_2$, and it is these quantities that must have positive real part for the filter condition to hold.

\begin{theorem}[Filter condition]\label{prop-hanoi1filter}
Under the above assumptions:
\begin{enumerate}[(a)]
\item When $r\geq \frac{3}{5}$ and $r\neq1$, there are solutions for $Z_1$ and $Z_2$, but the filter condition fails because $Z_2$ is purely imaginary for all $\omega \in \mathbb{R}_+$. 
\item If $r \in \bigl(0,\frac{3}{5}\bigr)$, the circuit is a filter precisely in the frequency range
\begin{gather*}
\gamma(r) - \sqrt{[\gamma(r)]^2 -1} < 2LC\omega^2 < \gamma(r) + \sqrt{[\gamma(r)]^2 -1}, \text{ where}\\
\gamma(r) = 1+\frac{r(3-5r)}{(2r-1)^2}.
\end{gather*}
\end{enumerate}
\end{theorem}

See Figure~\ref{fig:HanoiIFilter} for an interpretation of this theorem. Note that $\frac{3}{5}$ is the renormalization factor for fully symmetric resistances on the Hanoi graph to converge to a limiting resistance form on the fractal.  It seems unlikely that this value arises by coincidence in Theorem~\ref{prop-hanoi1filter}, because it also occurs in the condition for our second Hanoi circuit  (see Theorem~\ref{thm:HanoiIIfilter}). However, we do not have an explanation for its occurrence.

\begin{figure}[htp]
\centering
\includegraphics[width=0.7\textwidth]{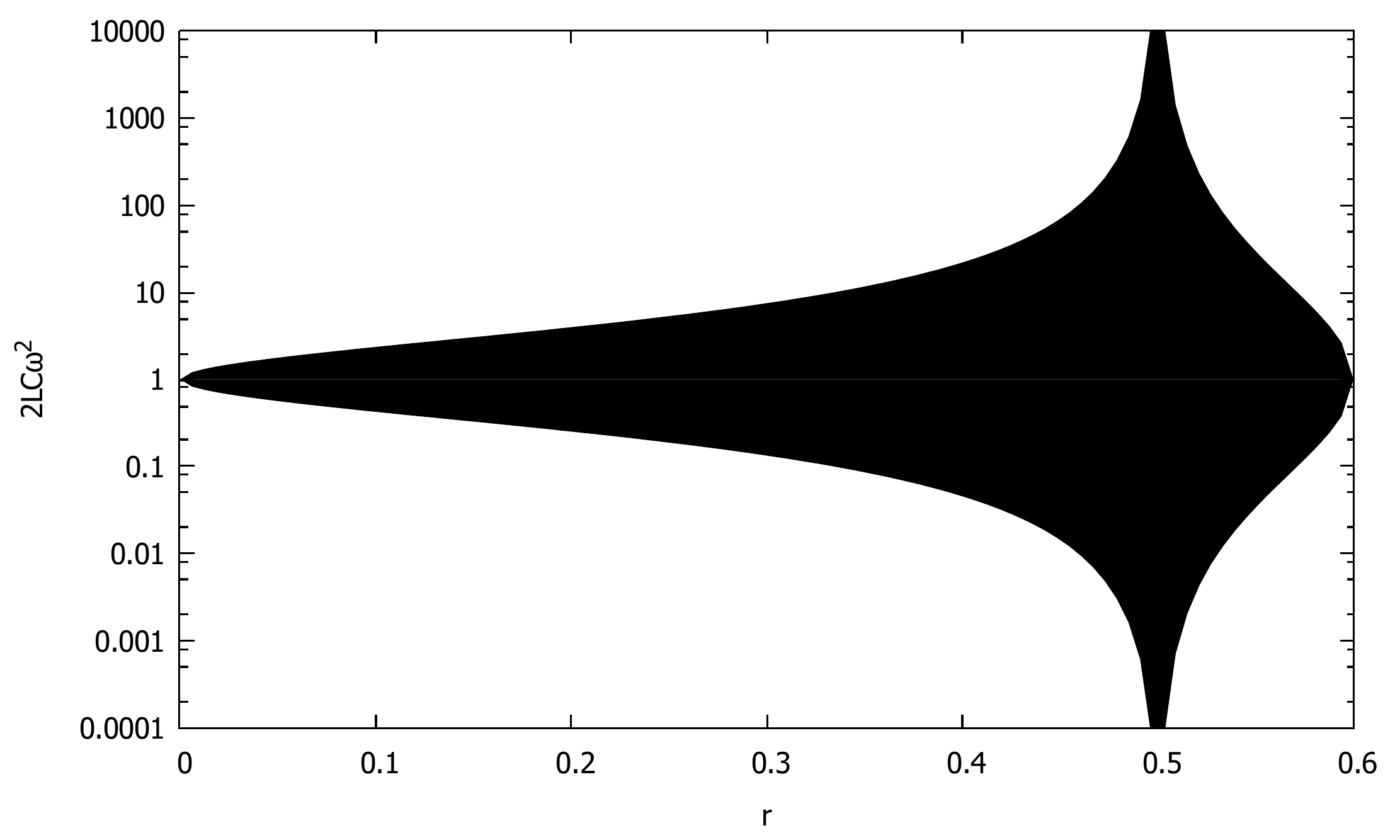}
\caption{The Hanoi Circuit I is a filter in the shaded region.}
\label{fig:HanoiIFilter}
\end{figure}

\begin{proof}
It is useful to rewrite~\eqref{eq:4} as
\begin{equation}\label{eq:5}
r(5r-3) Z_2^2 + (2r-1)\biggl(\frac{2}{i\omega C} +i\omega L\biggr) Z_2 + \frac{L}{C} =0,
\end{equation}
and note some special cases.  When $r\in \{0,\frac{3}{5}\}$, the quadratic term vanishes, so both $Z_2$ and $Z_1$ become purely imaginary for all $\omega$. When $r=\frac{1}{2}$, the linear term vanishes, and there is a positive root $Z_2=2\sqrt{\frac{L}{C}}$.

Assuming $r\notin \{ 0, \frac{3}{5}\}$, we may divide the quadratic equation by $r(5r-3)$ to get
\begin{equation*}
Z_2^2 + \frac{(2r-1)}{r(5r-3)} \biggl(\frac{2}{i\omega C} +i\omega L\biggr) Z_2 + \frac{1}{r(5r-3)}\frac{L}{C}=0,
\end{equation*}
whereupon writing $\alpha$ and $\beta$ for the roots we have
\begin{align*}
\alpha+\beta &= \frac{(2r-1)}{r(5r-3)} \biggl(\omega L - \frac{2}{\omega C}\biggr)i = a i, \\
\alpha\beta & = \frac{1}{r(5r-3)}\frac{L}{C}=b,
\end{align*}
where $a,b\in\mathbb{R}$.  Straightforward algebra shows that
\begin{equation*}
\alpha = \frac{ai + \sqrt{-a^2 -4b}}{2}, \quad \beta=\frac{ai - \sqrt{-a^2-4b}}{2},
\end{equation*}
so that the roots are purely imaginary when $-a^2-4b\leq 0$, or $b\geq -\frac{a^2}{4}$. Translating this result into our context, we deduce that both roots (i.e.\ both possible $Z_2$ values)
 are purely imaginary when
\begin{align}
\label{ineq}
\frac{1}{r(5r-3)}\frac{L}{C} \geq -\frac{1}{4} \left(\frac{2r-1}{r(5r-3)}\frac{L}{\omega}\right)^2 \left(\omega^2-\frac{2}{LC}\right)^2.
\end{align}
Observe that this is automatically satisfied when $r>\frac{3}{5}$ for \emph{any} $\omega$, since the LHS is positive and hence larger than the RHS. Combined with the case $r=\frac{3}{5}$ mentioned above, we deduce statement~(a) of the proposition. (Recall there is no solution when $r=1$.)

Turning to statement~(b), we first rearrange~\eqref{ineq} for  $r\in\bigl(0, \frac{3}{5}\bigr)\setminus\bigl\{\frac{1}{2}\bigr\}$ to read
\begin{equation*}
 LC\omega^2 \leq  \frac{(2r-1)^2}{4r(3-5r)} \bigl(LC\omega^2-2\bigr)^2,
\end{equation*}
which is the following quadratic in the dimensionless quantity $\Omega = LC\omega^2$:
\begin{equation*}
\Omega^2 - 4\left[1+\frac{r(3-5r)}{(2r-1)^2}\right]\Omega +4 \geq 0.
\end{equation*}
This inequality fails precisely when $LC\omega^2$ satisfies the inequality in statement~(b) of the proposition.  In this case, it is evident that one of the roots $\alpha,\beta$ has positive real part, hence there is $Z_2$ with positive real part.  Recall also that when $r=\frac{1}{2}$, the unique root  for $Z_2$ was positive.

Now rearranging the identity~\eqref{eq:3} we find that
\begin{equation*}
	(2-3r)(Z_1+Z_2) = (1-r)Z_2 + Z_C.
	\end{equation*}
Since $Z_C$ is imaginary, it follows that when $r\notin\bigl[\frac{2}{3}, 1\bigr]$, the real part of $Z_1+Z_2$ has the same sign as the real part of $Z_2$.  In particular, this is true when $r\in\bigl(0,\frac{3}{5}\bigr)$, on which interval we have already established the existence of $Z_2$ with positive real part.  Statement~(b) follows.
\end{proof}

\subsection*{Convergence of finite approximations}
Knowing this limiting circuit on the Hanoi attractor, one may ask whether the finite approximations to the Hanoi~I circuit converge as the level of approximation increases, or whether this would be the case if we were to put a small resistor in series with each circuit element.  Unfortunately we were not able to answer these questions.  We note, however, that the problem is more difficult than it was in the case of the Feynman-Sierpinski Ladder, because we have two characteristic resistances ($Z_1,Z_2$) rather than one, so the dynamical system will be two dimensional.  A second complicating factor is that one should presumably consider adding resistances to the finite approximations of the attractor as in Figure~\ref{fig-Hanoi}, rather than to the edges of the graphs in Figure~\ref{fig-Hanoi1}, which include vertices that do not exist in the limiting circuit.

 \subsection*{Harmonic functions}
As was true of the Feynman-Sierpinski Ladder,  the harmonic functions on the Hanoi circuit may be described by giving matrices that allow us to compute the values at vertices from level to level.   To establish notation, it is convenient to label the boundary vertices as $p_0,p_1,p_2$, and those linking the boundary arms to the central loop as $q_0,q_1,q_2$,  as shown in Figure~\ref{fig-Hanoi1impedances}.  The figure also shows the first step in reducing the circuit to the initial $Y$-configuration. 

\begin{figure}[htb]
  \centering
  \def\svgwidth{\columnwidth}
    \resizebox{0.9\textwidth}{!}{
\begingroup%
  \makeatletter%
  \providecommand\color[2][]{%
    \errmessage{(Inkscape) Color is used for the text in Inkscape, but the package 'color.sty' is not loaded}%
    \renewcommand\color[2][]{}%
  }%
  \providecommand\transparent[1]{%
    \errmessage{(Inkscape) Transparency is used (non-zero) for the text in Inkscape, but the package 'transparent.sty' is not loaded}%
    \renewcommand\transparent[1]{}%
  }%
  \providecommand\rotatebox[2]{#2}%
  \ifx\svgwidth\undefined%
    \setlength{\unitlength}{400.32460938bp}%
    \ifx\svgscale\undefined%
      \relax%
    \else%
      \setlength{\unitlength}{\unitlength * \real{\svgscale}}%
    \fi%
  \else%
    \setlength{\unitlength}{\svgwidth}%
  \fi%
  \global\let\svgwidth\undefined%
  \global\let\svgscale\undefined%
  \makeatother%
  \begin{picture}(1,0.36622975)%
    \put(0,0){\includegraphics[width=\unitlength,page=1]{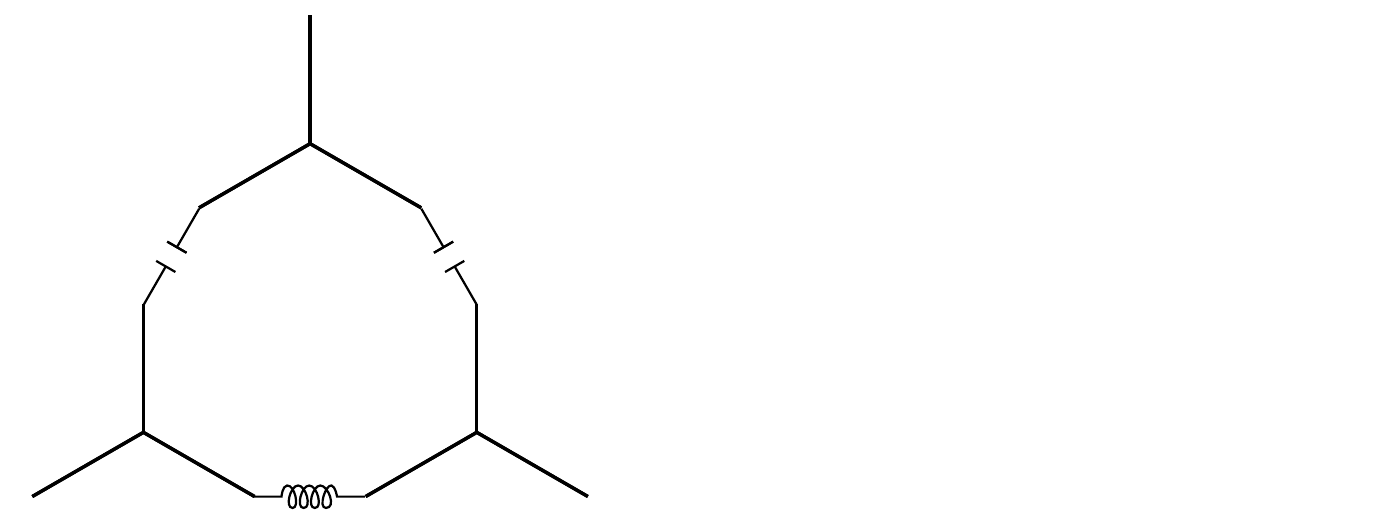}}%
    \put(0.2091057,0.22754054){\color[rgb]{0,0,0}\makebox(0,0)[lb]{\smash{$q_0$}}}%
    \put(0.11318354,0.06767028){\color[rgb]{0,0,0}\makebox(0,0)[lb]{\smash{$q_1$}}}%
    \put(0.30702623,0.06767028){\color[rgb]{0,0,0}\makebox(0,0)[lb]{\smash{$q_2$}}}%
    \put(0.23308624,0.34944161){\color[rgb]{0,0,0}\makebox(0,0)[lb]{\smash{$p_0$}}}%
    \put(-0.0027224,0.01571244){\color[rgb]{0,0,0}\makebox(0,0)[lb]{\smash{$p_1$}}}%
    \put(0.42792812,0.01571244){\color[rgb]{0,0,0}\makebox(0,0)[lb]{\smash{$p_2$}}}%
    \put(0.11,0.218){\color[rgb]{0,0,0}\makebox(0,0)[lb]{\smash{$p_{01}$}}}%
    \put(0.312,0.218){\color[rgb]{0,0,0}\makebox(0,0)[lb]{\smash{$p_{02}$}}}%
    \put(0.065,0.14){\color[rgb]{0,0,0}\makebox(0,0)[lb]{\smash{$p_{10}$}}}%
    \put(0.35,0.14){\color[rgb]{0,0,0}\makebox(0,0)[lb]{\smash{$p_{20}$}}}%
    \put(0.175,0.022){\color[rgb]{0,0,0}\makebox(0,0)[lb]{\smash{$p_{12}$}}}%
    \put(0.245,0.022){\color[rgb]{0,0,0}\makebox(0,0)[lb]{\smash{$p_{21}$}}}%
    \put(0,0){\includegraphics[width=\unitlength,page=2]{Hanoi1ImpedancesNew.pdf}}%
    \put(0.51885433,0.01571244){\color[rgb]{0,0,0}\makebox(0,0)[lb]{\smash{$p_1$}}}%
    \put(0.94850566,0.01571244){\color[rgb]{0,0,0}\makebox(0,0)[lb]{\smash{$p_2$}}}%
    \put(0.75666134,0.34944161){\color[rgb]{0,0,0}\makebox(0,0)[lb]{\smash{$p_0$}}}%
    \put(0.75566215,0.2585154){\color[rgb]{0,0,0}\makebox(0,0)[lb]{\smash{$q_0$}}}%
    \put(0.86857053,0.06966866){\color[rgb]{0,0,0}\makebox(0,0)[lb]{\smash{$q_2$}}}%
    \put(0.59379351,0.06966866){\color[rgb]{0,0,0}\makebox(0,0)[lb]{\smash{$q_1$}}}%
    \put(0.81061756,0.16758919){\color[rgb]{0,0,0}\makebox(0,0)[lb]{\smash{$rZ_1+rZ_2+Z_C$}}}%
    \put(0.69,0.02){\color[rgb]{0,0,0}\makebox(0,0)[lb]{\smash{$2rZ_2+Z_L$}}}%
  \end{picture}%
\endgroup%
}
\caption{Reduction of the Hanoi~I circuit.}\label{fig-Hanoi1impedances}
\end{figure}

\begin{lemma}\label{HanoiIharmoniconelevel}
If the values of a harmonic function at $p_0,p_1,p_2$ are expressed as a column vector, then the vector of values at $q_0,q_1,q_2$ may be obtained by multiplying by a matrix with eigenvalues $\lambda_j$ and eigenvectors $v_j$ as follows:
\begin{align*}
	\lambda_0 & = 1, && v_0=\begin{pmatrix}1&1&1\end{pmatrix}^T;\\
	\lambda_1&= (1-r), && v_-=\begin{pmatrix} 0&1&-1 \end{pmatrix}^T, \quad v_+=\begin{pmatrix}\frac{2Z_1}{Z_2}&-1&-1\end{pmatrix}^T.
\end{align*}
\end{lemma}
\begin{proof}
The first eigenvalue and eigenvector are obvious because the circuit is at a constant potential.  The symmetry in the circuit ensures the other eigenvectors are symmetric or antisymmetric under the reflection fixing $p_0$ and interchanging $p_1$ and $p_2$.  Considering the second circuit in Figure~\ref{fig-Hanoi1impedances}, we see that the antisymmetric eigenvector has voltage $(1-\lambda_1)$ across the arm from $p_1$ to $q_1$, $\lambda_1$ across the slanted sides of the triangle, and $2\lambda_1$ across the base of the triangle.  Summing currents at $q_1$ then yields
\begin{gather*}
	\frac{1-\lambda_1}{rZ_2} = \frac{\lambda_1}{rZ_1+rZ_2+Z_C} + \frac{2\lambda_1}{2rZ_2+Z_L},\text{ and hence}\\
	\frac{1}{rZ_2} =\lambda_1 \biggl( \frac{1}{rZ_2} +  \frac{1}{rZ_1+rZ_2+Z_C} + \frac{2}{2rZ_2+Z_L} \biggr)
			= \frac{\lambda_1}{r(1-r)Z_2},
	\end{gather*}
where the second equality is from~\eqref{eq:2}.  This gives $\lambda_1$.  For the symmetric one there is no voltage across the base of the triangle, so we may treat it by identifying $p_1$ with $p_2$ and $q_1$ with $q_2$ with edges in parallel.  Then the circuit is a line with impedances $rZ_1$, $(rZ_1+rZ_2+Z_C)/2$, and $rZ_2/2$, so the voltage from $q_0$ to $q_1$ and $q_2$ is proportional to that from $p_0$ to $p_1$ and $p_2$ with proportionality constant $(rZ_1+rZ_2+Z_C)/(3rZ_1+2rZ_2+Z_C)$. Since  $rZ_1+rZ_2+Z_C=(1-r)(2Z_1+Z_2)$ and $3rZ_1+2rZ_2+Z_C=2Z_1+Z_2$ by~\eqref{eq:1}, we get that the eigenvalue is again $(1-r)$.  Symmetry provides that the corresponding eigenvector is of the form $(a,-1,-1)$, so if we write the current through the line as $I$, we have voltage $IrZ_1=(1-\lambda_2)a$  from $p_0$ to $q_0$, and $IrZ_2/2=(1-\lambda_2)$ from $q_1$ to $p_1$.  Eliminating $I$ gives $a=2Z_1/Z_2$.
\end{proof}

The values at the six internal vertices on the loop in the left diagram of Figure~\ref{fig-Hanoi1impedances} are then obtained by linear interpolation according to the impedances in the loop.  In particular we have the following:

\begin{theorem}
Suppose the values of a harmonic function $h$ at the vertices $p_j$  are written using the vector $v=(h(p_0),h(p_1),h(p_2))^T$ as in Lemma~\ref{HanoiIharmoniconelevel}.  The values at the vertices of the three inverted $Y$-circuits at the first stage of the construction (Figure~\ref{fig-Hanoi1impedances} on left) are given by the vectors $(h(p_0),h(p_{01}),h(p_{02}))^T=M_0 v$, $(h(p_{10}),h(p_1),h(p_{12}))^T=M_1 v$, and $(h(p_{20}),h(p_{21}),h(p_2)^T=M_2 v$, where
\begin{align*}
	M_0
	&= \begin{bmatrix}
	1&0&0\\
	1-r & r(Z_1+Z_2)/b& rZ_1/b \\
	1-r & rZ_1/b & r(Z_1+Z_2)/b
	\end{bmatrix}, \\
	M_1
	&= \begin{bmatrix}
	r( Z_1+Z_2)/b & 1-r  & r/2 -rZ_2/2b \\
	0 & 1& 0 \\
	rZ_2/b & 1/2 - rZ_2/2b +c& 1/2 - rZ_2/2b -c
	\end{bmatrix},\\
	M_2
	&= \begin{bmatrix}
	r( Z_1+Z_2)/b & r/2 -rZ_2/2b & 1-r \\
	rZ_2/b & 1/2 - rZ_2/2b -c & 1/2 - rZ_2/2b +c\\
	0 & 0& 1
	\end{bmatrix},
	\end{align*}
and we have written $b=2Z_1+Z_2$ and $c= (1-r)Z_L/(4rZ_2+2Z_L)$ for notational simplicity.
\end{theorem}

\begin{corollary}
The eigenvalues and eigenvectors of the harmonic interpolation map $M_0$ are particularly simple: they are $1$, $r$, and $rZ_2/b$ with eigenvectors $v_0=(1,1,1)^T$, $v_1=(0,1,1)^T$,  and $v_2=(0,1,-1)^T$, respectively.
\end{corollary}

\begin{proof}
We write $h_\pm$ for the harmonic functions with boundary data $v_\pm$ as in Lemma~\ref{HanoiIharmoniconelevel}.  Recall that $h_-(q_0)=0$ and $h_-(q_1)=-h_-(q_2)=1-r$.  The values at $p_{01}$ and $p_{10}$ may be obtained by interpolating using the series of impedances $rZ_2$,  $Z_C$, $rZ_1$,  normalized by dividing by $rZ_1+rZ_2+Z_C$. Since $rZ_1+rZ_2+Z_C=(1-r)b$ by~\eqref{eq:1}, we obtain $h_-(p_{01})=-h_-(p_{02})=rZ_2/b$ and $h_-(p_{10})=-h_-(p_{20})=(rZ_2+Z_C)/b$.  Similarly we may interpolate between the values at $q_1$ and $q_2$ with weights $rZ_2$, $Z_L$ and $rZ_2$ normalized by $(2rZ_2+Z_L)^{-1}$.  The voltage drop from $q_1$ to $p_{12}$ is $2r(1-r)Z_2/(2rZ_2+Z_L)$, so after a little algebra $h_-(p_{12})=-h_-(p_{21})=(1-r)Z_L/(2rZ_2+Z_L)=2c$.

Similar reasoning applies to $h_+$.  In the proof of Lemma~\ref{HanoiIharmoniconelevel} we identified vertices and treated the circuit as a line.  For the segment between $q_0$ and $q_1,q_2$, the impedances are $rZ_2/2$, $Z_C/2$ and $rZ_1/2$, normalized by a factor $2/(rZ_1+rZ_2+Z_C)=2/((1-r)b)$.  The voltage difference is $(1-r)b/Z_2$, so the voltage differences are $r$, $Z_C/Z_2$ and $rZ_1/Z_2$.  Thus $h_+(p_{01})=h_+(p_{02})= 2(1-r)Z_1/Z_2 -r$ and $h_+(p_{10})=h_+(p_{20})=r-1+rZ_1/Z_2$.  Across the base there is no voltage difference, so $h_+(p_{12})=h_+(p_{21})=h_+(q_1)=r-1$.

The columns of the matrices are obtained using the values of the functions $g_0 =Z_2(1+h_+)/b$ and $g_\pm=(2Z_1-Z_2 h_+)/(2b)\pm h_-/2$ at the vertices of the $Y$-circuit corresponding to the matrix.  We compute
\begin{align*}
	M_0
	&=\begin{bmatrix}
	g_0(p_0) & g_+(p_0) & g_-(p_0) \\
	g_0(p_{01}) & g_+(p_{01}) & g_-(p_{01}) \\
	g_0(p_{02}) & g_+(p_{02}) & g_-(p_{02}) 
	\end{bmatrix}
	= \begin{bmatrix}
	1&0&0\\
	1-r & r(Z_1+Z_2)/b& rZ_1/b \\
	1-r & rZ_1/b & r(Z_1+Z_2)/b
	\end{bmatrix},
	\\
	M_1
	&=\begin{bmatrix}
	g_0(p_{10}) & g_+(p_{10}) & g_-(p_{10}) \\
	g_0(p_1) & g_+(p_1) & g_-(p_1) \\
	g_0(p_{12}) & g_+(p_{12}) & g_-(p_{12}) 
	\end{bmatrix}
	= \begin{bmatrix}
	r( Z_1+Z_2)/b & 1-r  & r/2 -rZ_2/2b \\
	0 & 1& 0 \\
	rZ_2/b & 1/2 - rZ_2/2b +c& 1/2 - rZ_2/2b -c
	\end{bmatrix},\\
	M_2
	&=\begin{bmatrix}
	g_0(p_{20}) & g_+(p_{20}) & g_-(p_{20}) \\
	g_0(p_{21}) & g_+(p_{21}) & g_-(p_{21}) \\
	g_0(p_2) & g_+(p_2) & g_-(p_2)
	\end{bmatrix}
	= \begin{bmatrix}
	r( Z_1+Z_2)/b & r/2 -rZ_2/2b & 1-r \\
	rZ_2/b & 1/2 - rZ_2/2b -c & 1/2 - rZ_2/2b +c\\
	0 & 0& 1
	\end{bmatrix}.
	\end{align*}\qedhere
\end{proof}
% let b=(2Z_1+Z_2)

%g_0 = 
% g_0(p_10)=r( Z_1+Z_2)/b
% g_0(p_12) = rZ_2/b

% 1+h_+ is (1-r)(2Z_1/Z_2 +1)=2b/Z_2 at p_01 and r(1+Z_1/Z_2) at p_10 and r at p_12.  then rescale by Z_2/b to  get results above

% 2Z_1-Z_2h_+ is rb at p_01 and b -r(Z_1+Z_2) at p_10  and b - rZ_2 at p_12. divide by 2b to get r/2 at p_01, 1/2 - r(Z_1+Z_2)/2b at p_10 and 1/2 - rZ_2/2b at p_12.    (consistent with being half of (1,1,1)-(1,0,0))

% when add h_-/2  to the above get  r/2 + rZ_2/2b= (rb+rZ_2)/2b=r(Z_1+Z_2)/b at p_01 and r/2 - rZ_2/2b=rZ_1/b at p_02.  Also  1/2 + (Z_C-rZ_1)/2b=1/2 + (Z_C+rZ_1+rZ_2 - rb)/2b = 1/2 + ((1-r)b -rb)/2b = 1 -r at p_10 and 1/2 - (rZ_1+2rZ_2+Z_C)/2b =r/2 -rZ_2/2b at p_20.  And 1/2 - rZ_2/2b \pm (1-r)Z_L/(2rZ_2+Z_L)  at p_12 and p_21

% when subtract h_-/2 get r/2-rZ_2/b at p_01 and r/2+rZ_2/b at p_02.  Also 1/2 - 

%g_+(p_01) = rZ_1/2(2Z_1+Z_2)
%g_+(p_02) = r(z_1+2Z_2)/(2Z_1+Z_2)
%g_+(p_10)= 

%trace is (1+r)/2 - c     
% det is  

\section{Hanoi Circuit~II}\label{sec-Hanoi2}
A variation of the Hanoi~I circuit is shown in Figure~\ref{fig-Hanoi2} below.  Again we begin with a  bilaterally symmetric inverted Y  circuit in which the vertical arm has impedance $Z_1$ and the angled arms each have impedance $Z_2$, as shown on the left in the figure.  Three copies of this with impedances scaled by $r>0$ are glued as shown on the right in the figure, and the process is repeated indefinitely.
\begin{figure}[htb]
  \centering
  \def\svgwidth{\columnwidth}
    \resizebox{0.8\textwidth}{!}{
		\begingroup%
  \makeatletter%
  \providecommand\color[2][]{%
    \errmessage{(Inkscape) Color is used for the text in Inkscape, but the package 'color.sty' is not loaded}%
    \renewcommand\color[2][]{}%
  }%
  \providecommand\transparent[1]{%
    \errmessage{(Inkscape) Transparency is used (non-zero) for the text in Inkscape, but the package 'transparent.sty' is not loaded}%
    \renewcommand\transparent[1]{}%
  }%
  \providecommand\rotatebox[2]{#2}%
  \ifx\svgwidth\undefined%
    \setlength{\unitlength}{364.88687299bp}%
    \ifx\svgscale\undefined%
      \relax%
    \else%
      \setlength{\unitlength}{\unitlength * \real{\svgscale}}%
    \fi%
  \else%
    \setlength{\unitlength}{\svgwidth}%
  \fi%
  \global\let\svgwidth\undefined%
  \global\let\svgscale\undefined%
  \makeatother%
  \begin{picture}(1,0.41999649)%
    \put(0,0){\includegraphics[width=\unitlength,page=1]{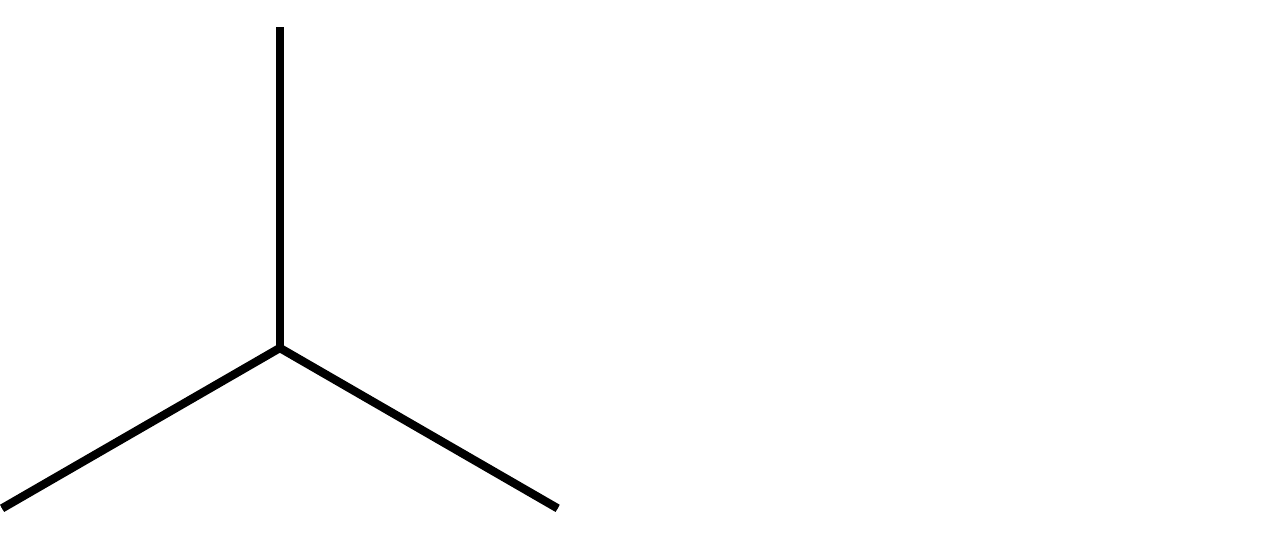}}%
    \put(0.23697609,0.27){\color[rgb]{0,0,0}\makebox(0,0)[lb]{\smash{$Z_1$}}}%
    \put(0.34,0.1){\color[rgb]{0,0,0}\makebox(0,0)[lb]{\smash{$Z_2$}}}%
    \put(0.086,0.1){\color[rgb]{0,0,0}\makebox(0,0)[lb]{\smash{$Z_2$}}}%
    \put(0,0){\includegraphics[width=\unitlength,page=2]{Hanoi2inkscape.pdf}}%
    \put(0.78195445,0.30923964){\color[rgb]{0,0,0}\makebox(0,0)[lb]{\smash{$rZ_1$}}}%
    \put(0.73,0.342){\color[rgb]{0,0,0}\makebox(0,0)[lb]{\smash{$p_{00}$}}}%
    \put(0.69,0.26){\color[rgb]{0,0,0}\makebox(0,0)[lb]{\smash{$rZ_2$}}}%
    \put(0.704,0.217){\color[rgb]{0,0,0}\makebox(0,0)[lb]{\smash{$p_{01}$}}}%
    \put(0.805,0.217){\color[rgb]{0,0,0}\makebox(0,0)[lb]{\smash{$p_{02}$}}}%
    \put(0.86961396,0.20562622){\color[rgb]{0,0,0}\makebox(0,0)[lb]{\smash{$Z_C$}}}%
    \put(0.64,0.20562622){\color[rgb]{0,0,0}\makebox(0,0)[lb]{\smash{$Z_C$}}}%
    \put(0.674,0.166){\color[rgb]{0,0,0}\makebox(0,0)[lb]{\smash{$p_{10}$}}}%
    \put(0.836,0.166){\color[rgb]{0,0,0}\makebox(0,0)[lb]{\smash{$p_{20}$}}}%
    \put(0.595,0.035){\color[rgb]{0,0,0}\makebox(0,0)[lb]{\smash{$p_{11}$}}}%
    \put(0.715,0.035){\color[rgb]{0,0,0}\makebox(0,0)[lb]{\smash{$p_{12}$}}}%
    \put(0.8,0.035){\color[rgb]{0,0,0}\makebox(0,0)[lb]{\smash{$p_{21}$}}}%
    \put(0.91,0.035){\color[rgb]{0,0,0}\makebox(0,0)[lb]{\smash{$p_{22}$}}}%
    \put(0.75715609,0.01039593){\color[rgb]{0,0,0}\makebox(0,0)[lb]{\smash{$Z_L$}}}%
    \put(0,0){\includegraphics[width=\unitlength,page=3]{Hanoi2inkscape.pdf}}%
    \put(0.78501752,0.3744849){\color[rgb]{0,0,0}\makebox(0,0)[lb]{\smash{$Z_L$}}}%
    \put(0.96,0.05){\color[rgb]{0,0,0}\makebox(0,0)[lb]{\smash{$Z_C$}}}%
    \put(0.54,0.05){\color[rgb]{0,0,0}\makebox(0,0)[lb]{\smash{$Z_C$}}}%
    \put(0.81,0.26){\color[rgb]{0,0,0}\makebox(0,0)[lb]{\smash{$rZ_2$}}}%
  \end{picture}%
\endgroup%
}
\caption{Variant Hanoi-type circuit construction.}\label{fig-Hanoi2}
\end{figure} 

As with our first Hanoi circuit, we compute the impedance from the top vertex to the bottom pair of vertices, and between the two bottom vertices, and equate the results for the left and right diagrams.  The result is the two equations
\begin{gather}
Z_1 +\frac{1}{2} Z_2 = Z_L + r Z_1 + \frac{1}{2}( r Z_1 + 2r Z_2 + 2 Z_C ), \quad\text{and}\label{C1}\\
2 Z_2 = 2 r Z_2 + 2 Z_C + \bigg( \frac{1}{2rZ_2 + Z_L} + \frac{1}{2r ( Z_1 + Z_2) + 2Z_C} \bigg)^{-1}. \label{C2}
\end{gather}
We may reorganize this to get a quadratic for $Z_2$ and a formula for $Z_1$ in terms of $Z_2$.

\begin{lemma}
If $r\neq\frac{2}{3}$, then
\begin{align}
	Z_1= \frac{2r-1}{2-3r} Z_2 + \frac{2}{2 - 3r}(Z_C + Z_L), \quad\text{and} \label{C3}\\
	Z_2^2 \bigl( 2r(1-r)(5r-3)\bigr) +2 Z_2 \bigl( 2(1-r)(3r-1)Z_C +& (2r-1)(r+1)Z_L \bigr) \notag\\ &\quad + 2(Z_L+Z_C)\bigl( (2-r)Z_C+rZ_L  \bigr) =0; \label{C6}
	\end{align}
while if $r=\frac{2}{3}$, then one can  solve explicitly for $Z_1$ and $Z_2$, and the impedances are purely imaginary.
\end{lemma}
\begin{proof}
With $r\neq\frac{2}{3}$, we obtain~\eqref{C3} from~\eqref{C1} by elementary algebra.  It is also convenient to compute
\begin{equation}\label{C5}
	Z_1 +Z_2= \frac{1-r}{2-3r} Z_2 + \frac{2}{2 - 3r}(Z_C + Z_L).
	\end{equation}
Using~\eqref{C5} to replace $Z_1+Z_2$ in the following version of~\eqref{C2}:
\begin{equation*}
	\frac{1}{2(1-r)Z_2 -2Z_C} = \frac{1}{2rZ_2 + Z_L} + \frac{1}{2r ( Z_1 + Z_2) + 2Z_C},
	\end{equation*}
we obtain an expression that reduces to the given quadratic.

In the case $r=\frac{2}{3}$, we instead obtain $Z_2=-6(Z_C+Z_L)$ from~\eqref{C1}, and using~\eqref{C5} then gives
\begin{equation*}
	\frac{1}{18Z_C+12Z_L} = \frac{1}{24Z_C+21Z_L} + \frac{1}{18Z_C+24Z_L-4Z_1},
	\end{equation*}
from which $Z_1=-3(Z_C+Z_L)(9Z_C+Z_L)/(2Z_C+3Z_L)$ is purely imaginary.
\end{proof}

\begin{theorem}\label{thm:HanoiIIfilter}
For $r\in\bigl(0,\frac{1}{2}\bigr)\cup\bigl(\frac{1}{2},\frac{3}{5}\bigr)$,  there are frequencies for which the Hanoi~II circuit  is a filter.  With $Z_C=\frac{1}{i\omega C}$ and $Z_L=i\omega L$, the frequency range is given by the following condition on $\omega^2 LC$:
\begin{equation}\label{Hanoi2conditforfilter}
	\begin{split} \bigl(-24r^4+28r^3-9r^2+2r-1\bigr)(\omega^2LC)^2  &-4(r-1)^2\bigl( 6r^2-3r-1\bigr)(\omega^2LC) \\
	&\quad +4(1-r)\bigl( 4r^3-2r^2+r-1  \bigr) >0, \end{split}
	\end{equation}
which is graphed in Figure~\ref{Hanoi2filtergraph}.  In particular, if $\omega^2 LC=2$, then~\eqref{Hanoi2conditforfilter} is valid for all $r\in\bigl(0,\frac35\bigr)\setminus\bigl\{\frac{1}{2}\bigr\}$.
\end{theorem}

\begin{proof}
As was the case for the Hanoi~I circuit, it is necessary and sufficient that both $Z_2$ and $Z_1+Z_2$ have positive real part.  Moreover, from~\eqref{C5} and the fact that $Z_C+Z_L$ is imaginary, it is apparent that the signs of the real parts of  $Z_2$ and $Z_1+Z_2$ are the same if and only if $r\not\in\bigl[\frac{2}{3},1\bigr]$.

The quadratic~\eqref{C6} for $Z_2$ is of the form $aZ_2^2 + b i Z_2 + c =0$  for some real numbers $a,b,c$.  This has a root with non-zero real part if and only if $-b^2-4ac>0$.  Substituting for $a,b,c$ gives a condition that is quartic in $r$ but only  quadratic in $\omega^2 LC$.  Specifically
\begin{align*}
	-b^2	&= \frac{-4}{\omega^2 C^2} \bigl( - 2(1-r)(3r-1)+(2r-1)(r+1)\omega^2 LC \bigr),\\
	-4ac 	&= \frac{16}{\omega^2 C^2} r(1-r)(5r-3) (\omega^2 LC - 1)(r\omega^2 LC+r-2),
	\end{align*}
so that $-b^2-4ac>0$ precisely when~\eqref{Hanoi2conditforfilter} holds. (Note that $\omega^2 LC\geq0$.)

In the expression on the left of~\eqref{Hanoi2conditforfilter}, both the constant coefficient and the coefficient of $(\omega^2 LC)^2$ are negative for all $r\not\in\left[\frac{2}{3},1\right]$,  while the coefficient of $\omega^2 LC$ is positive only when $(1-\sqrt{11/3})<4r<(1+\sqrt{11/3})$. The latter interval does not intersect $(1,\infty)$, so when $r>1$ there are no frequencies for which~\eqref{Hanoi2conditforfilter} holds, and the  circuit cannot be a filter.  However this interval contains $(0,2/3)$.  It follows from these observations that for $r\in(0,2/3)$ the expression has a maximum at a positive value of $\omega^2 LC$, and the maximum is
\begin{equation*}	
	4(1-r)( 4r^3-2r^2+r-1 )  - \frac{4(r-1)^4( 6r^2-3r-1)^2}{(-24r^4+28r^3-9r^2+2r-1)},
	\end{equation*}
which is positive for an $r\in(0,2/3)$ if and only if
\begin{align*}
	0 &> ( 4r^3-2r^2+r-1 ) (-24r^4+28r^3-9r^2+2r-1)  + (r-1)^3 ( 6r^2-3r-1)^2\\
	&= -4r(3r-2)(5r-3)(r+1)^2(r-\frac{1}{2})^2,
	\end{align*}
from which we determine that $r\in(0,\frac{3}{5})$ and $r\neq\frac{1}{2}$ is necessary and sufficient for the existence of a frequency $\omega$ such that the filter condition applies.  To see that one such frequency is $\omega^2 LC=2$, one may substitute this value into~\eqref{Hanoi2conditforfilter}, at which point the inequality reduces to $-8r(2r-1)^2 (5r-3)>0$, which is true for the given values of $r$.
\end{proof}

\begin{figure}[htb]
\centering
\includegraphics[width=0.7\textwidth]{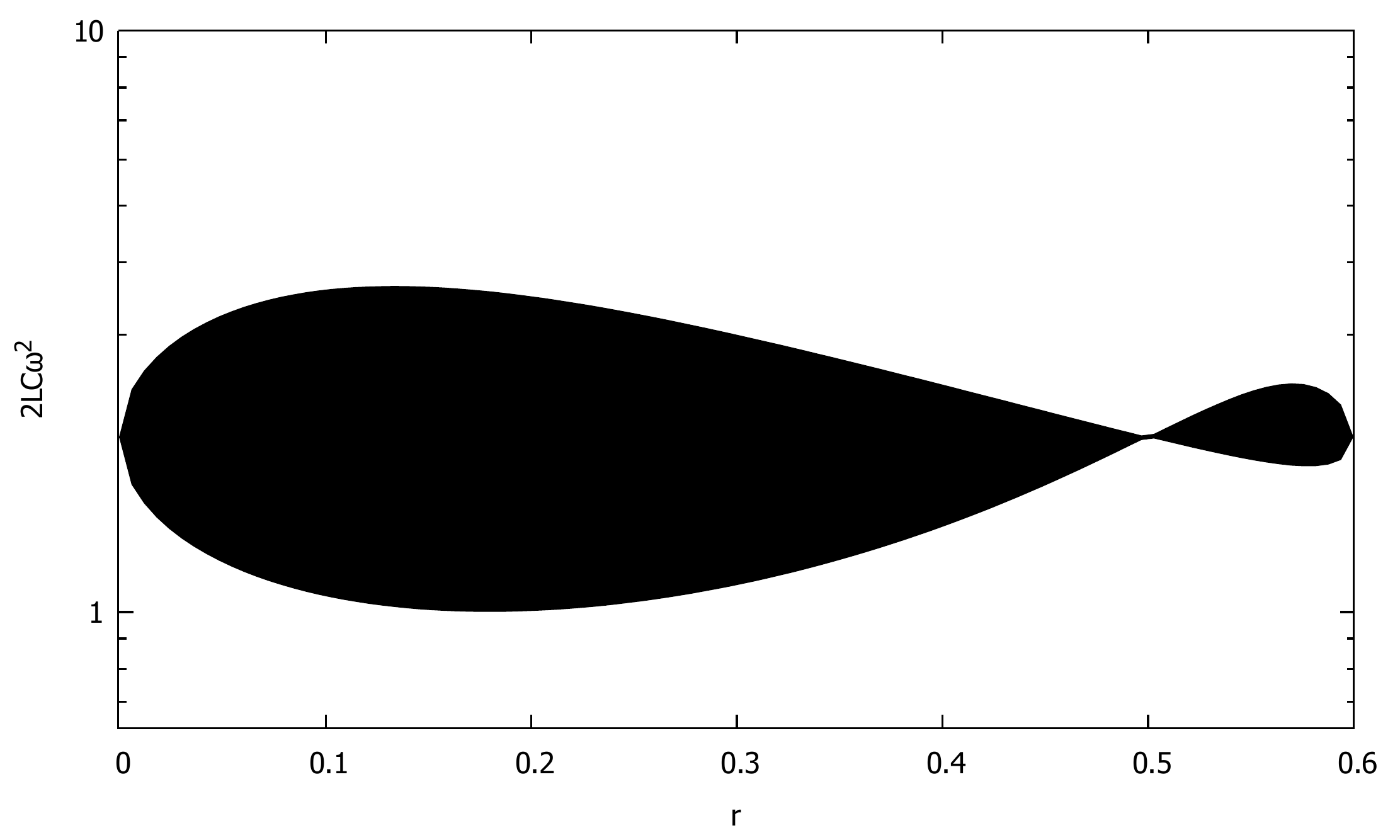}
\caption{The Hanoi Circuit II is a filter in the shaded region.}\label{Hanoi2filtergraph}
\end{figure}

\subsection*{Convergence of finite approximations}
We do not know whether the sequence of finite approximating circuits, or a sequence of circuits in which there are small resistors in series with the imaginary impedances, converges to our circuit on the Hanoi attractor.  The same issues that made this difficult in the case of the Hanoi~I circuit arise in the case of the Hanoi~II circuit.

\subsection*{Harmonic functions}
We may find matrices  that can be used to compute harmonic functions in a manner similar to that used for the Hanoi~I circuit. For this purpose, we label the vertices of the graph in the same way as was used for the Hanoi~I circuit in Figure~\ref{fig-Hanoi1impedances},  so the external vertices are $p_0,p_1,p_2$, and the central vertices of each copy of the inverted Y-circuit are $q_0,q_1,q_2$.

\begin{lemma}\label{lemma:hanoiharmonic} If the values of a harmonic function at $p_0, p_1, p_2$ are expressed as a column vector, then the vector of values at $q_0, q_1, q_2$ may be obtained by multiplying by a matrix with eigenvalues $\lambda_j$ and eigenvectors $v_j$ as follows:  
\begin{align*}
	\lambda_0 & = 1, && v_0 = \begin{pmatrix} 1 & 1 & 1 \end{pmatrix}^T;\\
	\lambda_- &= 1-r- \frac{Z_C}{Z_2},  &&  v_- = \begin{pmatrix} 0 & 1 & -1 \end{pmatrix}^T;\\
	\lambda_+ &=\frac{rZ_1 + rZ_2 +Z_C}{ 2Z_1 + Z_2},    &&  v_+ = \begin{pmatrix} 2 (rZ_1 + Z_L)/(rZ_2 +Z_C) & -1 & -1 \end{pmatrix}^T.
	\end{align*}
\end{lemma}
\begin{proof}
We proceed as in Lemma~\ref{HanoiIharmoniconelevel}. The first eigenvalue and eigenvector are obvious, and we need only consider eigenvectors that are symmetric and antisymmetric under the reflection fixing $p_0$ and exchanging $p_1$ and $p_2$.  The antisymmetric eigenvector $v_-$ has voltage $(1-\lambda_-)$ across the arm from $p_1$ to $q_1$, $\lambda_-$ across the slanted sides of the triangle, and $2 \lambda_-$ across the base of the triangle.  Summing currents at $q_1$ yields
\begin{align*}
\frac{1-\lambda_-}{rZ_2 + Z_C} &= \frac{\lambda_-}{rZ_1 + rZ_2 + Z_C} + \frac{2\lambda_-}{2rZ_2 + Z_L}, \\
\frac{1}{rZ_2 + Z_C} &= \lambda_- \bigg( \frac{1}{rZ_2 +Z_C} + \frac{1}{(1-r)Z_2 -Z_C} \bigg), 
\end{align*}
where the second equation is obtained by applying~\eqref{C2} to the first.  Thus $\lambda_- = 1-r- \frac{Z_C}{Z_2}$.  For the symmetric eigenvector, $p_1$ may be identified with $p_2$, and $q_1$ with $q_2$.  Thus our circuit becomes a line with impedances $rZ_1 + Z_L$, $(rZ_1 +rZ_2+Z_C)/2$ and $(rZ_2 +Z_C)/2$.  The voltage from $q_0$ to the pair $q_1,q_2$ is proportional to that from $p_0$ to $p_1,p_2$ with proportionality constant 
\begin{equation*}
	\frac{rZ_1 + rZ_2 +Z_C}{3rZ_1 + 2rZ_2 + 2Z_C + 2Z_L}= \frac{rZ_1 + rZ_2 +Z_C}{ 2Z_1 + Z_2} = \lambda_+,
	\end{equation*}
where we simplified using~\eqref{C1}.  The corresponding eigenvector is of the form $(a,-1,-1)$, so if we write the current through the line as $I$, we have voltage $I (rZ_1 + Z_L) = (1-\lambda_+) a$ from $p_0$ to $q_0$, and $I (rZ_2 +Z_C)/2 = (1-\lambda_+)$ from $q_1$ to $p_1$.  Eliminating $I$ gives $a = 2 (rZ_1 + Z_L)/(rZ_2 +Z_C)$.  
\end{proof}

By computing values of harmonic functions using this lemma, we can then compute the harmonic interpolation matrices for the Hanoi~II circuit.

\begin{theorem}
Let $M_j$ denote the harmonic interpolation matrix from the values at $(p_0,p_1,p_2)$ to those at $(p_{j0},p_{j1},p_{j2})$, with point labelled as in Figure~\ref{fig-Hanoi2}. Then
\begin{align*}
	M_0
	&= \begin{bmatrix}
	1-2Z_L/b& Z_L/b & Z_L/b\\
	(c+Z_C)/b&r/2+Z_L/b + rd/2c&r/2+Z_L/b - rd/2c \\
	(c+Z_C)/b&r/2+Z_L/b-  rd/2c&r/2+Z_L/b + rd/2c
	\end{bmatrix},
	\\
	M_1
	&= \begin{bmatrix}
	c/b&c/2b+d(Z_2-d)/2cZ_2 &  c/2b- d(Z_2-d)/2cZ_2\\
	Z_C/b & 1-Z_C(Z_1+Z_2)/bZ_2& Z_1Z_C/bZ_2 \\
	(rZ_2+Z_C)/b& (b+d-Z_2)/2b + dZ_L/2Z_2(2rZ_2+Z_L) & (b+d -Z_2)/2b - dZ_L/2Z_2(2rZ_2+Z_L) 
	\end{bmatrix},\\
	M_2
	&= \begin{bmatrix}
	c/b& c/2b- d(Z_2-d)/2cZ_2& c/2b+d(Z_2-d)/2cZ_2\\
	(rZ_2+Z_C)/b& (b+d -Z_2)/2b - dZ_L/2Z_2(2rZ_2+Z_L) &(b+d-Z_2)/2b + dZ_L/2Z_2(2rZ_2+Z_L)\\
	Z_C/b & Z_1Z_C/bZ_2  &  1-Z_C(Z_1+Z_2)/bZ_2
	\end{bmatrix},
	\end{align*}
where $b=2Z_1+Z_2$,  $c=rZ_1+rZ_2+Z_C$, and $d=(1-r)Z_2-Z_C$.
\end{theorem}
\begin{proof}
As in Lemma~\ref{lemma:hanoiharmonic} we define harmonic functions $h_\pm$ with boundary data $v_\pm$. 
Recall that for the symmetric eigenfunction $h_+$ the circuit is equivalent to a line segment. Consider $g_0=(1+h_+)/(a+1)$ the harmonic function with boundary values $(1,0,0)^T$, where  $a=2(rZ_1+Z_L)/(rZ_2+Z_C)$ as in the proof of Lemma~\ref{lemma:hanoiharmonic}.  Since the values of $g_0$ go from $0$ to $1$ along the line, the current is the reciprocal of the total resistance $(2Z_1+Z_2)/2$, and we can compute the voltages directly from this and the resistances.  Putting $b=2Z_1+Z_2$ as before this gives $g_0(p_{11})=g_0(p_{22})=Z_c/b$, $g_0(p_{12})=g_0(p_{21})=g_0(q_1)=(rZ_2+Z_C)/b=(Z_2-d)/b$, $g_0(p_{10})=g_0(p_{20})=c/b$, $g_0(p_{01})=g_0(p_{02})=(c+Z_C)/b$, $g_0(p_{00})=1-2Z_L/b$.

Now let $g_\pm= \frac{1}{2}(1-g_0 \pm h_-)$.  We need the values of $h_-$ at the $p_{jk}$ vertices.  Symmetry gives $h_-(p_{00})=h_-(q_0)=0$, and we know $h_-(q_1)=-h_-(q_2)=d/Z_2$.  Then $h_-(p_{01})=-h_-(p_{02})= rd/c$, $h_-(p_{10})=-h_-(p_{20})= d(Z_2-d)/cZ_2$, $h_-(p_{11})=-h_-(p_{22})= 1-Z_C/Z_2$, and $h_-(p_{12}) = -h_-(p_{21})= dZ_L/Z_2(2rZ_2+Z_L)$.  From these and the values of $g_0$ we can fill the matrix entries in the following expressions to complete the proof.
\begin{align*}
	M_0
	&=\begin{bmatrix}
	g_0(p_{00}) & g_+(p_{00}) & g_-(p_{00}) \\
	g_0(p_{01}) & g_+(p_{01}) & g_-(p_{01}) \\
	g_0(p_{02}) & g_+(p_{02}) & g_-(p_{02}) 
	\end{bmatrix},\\
	M_1
	&=\begin{bmatrix}
	g_0(p_{10}) & g_+(p_{10}) & g_-(p_{10}) \\
	g_0(p_{11}) & g_+(p_{11}) & g_-(p_{11}) \\
	g_0(p_{12}) & g_+(p_{12}) & g_-(p_{12}) 
	\end{bmatrix},\\
	M_2
	&=\begin{bmatrix}
	g_0(p_{20}) & g_+(p_{20}) & g_-(p_{20}) \\
	g_0(p_{21}) & g_+(p_{21}) & g_-(p_{21}) \\
	g_0(p_{22}) & g_+(p_{22}) & g_-(p_{22})
	\end{bmatrix}.
	\end{align*}\qedhere
\end{proof}

\bibliographystyle{plain}
\bibliography{AC}

\vspace{20pt}

\end{document}